\documentclass[11pt]{article}
\usepackage[utf8]{inputenc}
\usepackage{csquotes}
\usepackage{setspace}
\usepackage[dvipsnames]{xcolor}
\usepackage[margin=1in]{geometry}
\usepackage{hyperref}
\usepackage{amsmath, amssymb, amscd, amsthm, amsfonts}
\usepackage{graphicx}
\usepackage{adjustbox}
\usepackage{hyperref}
\usepackage{cleveref}
\usepackage{amsthm}
\usepackage[utf8]{inputenc}
\usepackage[english]{babel}
\usepackage{proof}
\usepackage{ninecolors}
\usepackage{enumitem}

\usepackage{proof}
\usepackage{tikz-cd}
\usepackage{xcolor}
\usepackage{thmtools, thm-restate}
\tikzcdset{scale cd/.style={every label/.append style={scale=#1},
		cells={nodes={scale=#1}}}}

\newtheorem{theorem}{Theorem}[section]

\theoremstyle{definition}
\newtheorem{definition}{Definition}[section]

\title{Simplicial Actions for Distributed Protocols}
\author{Philip Sink}
\date{}

\newcommand{\M}{\mathcal{M}}

\newcommand{\commentout}[1]{}
\newcommand{\defin}[1]{\textbf{#1}}

\renewcommand{\phi}{\varphi}

\newcommand{\lthen}{\rightarrow}

\newcommand{\F}{\mathcal{F}}
\renewcommand{\M}{\mathcal{M}}
\renewcommand{\L}{\mathcal{L}}

\usepackage{amssymb,tikz}

\begin{document}
	
	\maketitle
	
\begin{abstract}
	
	This paper captures and extends some of the core results from the tech memo ``A New Semantics for Belief Revision in Simplicial Complexes''. \cite{Memo} As such, we set out to explore the implementation of action models in the setting of simplicial semantics for modal logic. Such an idea is not entirely new to the literature, showing up in both \cite{SimpDEL} and \cite{KaSC}. However, we will explore action models in a more general setting. In particular, we will allow for action models for simplicial models for \textit{belief}, as in \cite{BelSimp}. This will let us incorporate the notion of belief revision, as developed in \cite{SimpBelRev}, into these action models. Moreover, we explicitly connect action models in the simplicial setting to distributed protocols as defined in \cite{DCTCT}. We conclude with some speculation on how we might interpret distributed protocols with revision.
	
\end{abstract}

\section{Introduction} \label{3sec:int}

Seeing as the impetus for Simplicial Semantics is the use of combinatorial topology in distributed computing, most of that literature acknowledges and is indebted to this connection already. \cite{Death,DoA,KaSC,SimpDEL,HG,FA,DCTCT} In general, the procedures for using action models to describe distributed protocols are understood but new, showing up in \cite{SimpDEL}, \cite{KaSC}, \cite{KaG}, \cite{ActionDC}, \cite{CPM}, and \cite{CPaAM}. Moreover, the papers \cite{KaSC}, \cite{BelSimp}, and \cite{SimpBel} are already attempts to give a semantics for belief, or non-factive (defeasible) knowledge. Our paper will seek to build on this work by implementing simplicial actions in the setting of belief models. In doing so, we will be able to describe actions where the agents update via belief revision as described in \cite{SimpBelRev}. Moreover, we will explicitly connect these action models to distributed protocols as defined in \cite{DCTCT}. 

The paper will proceed as follows. In Section \ref{3sec:simpsem}, we will outline the basic definitions for simplicial semantics used throughout the paper. In Section \ref{3sec:SimpProt}, we explore the basics of connecting distributed protocols to combinatorial topology. In Section \ref{3sec:ActProt}, we explore how action models in the non-simplicial setting have been used to model protocols. In Section \ref{3sec:SimpAct}, we define action models for both simplicial models of knowledge and belief, and also define the latter where agents update via belief revision. We also give examples of some existing protocols and show how to model them using simplicial actions.

\section{Simplicial Semantics}\label{3sec:simpsem}

This section will give the particulars for how we will define simplicial semantics throughout the paper. For a more thorough overview, we recommend reading \cite{KaSC}, \cite{BelSimp}, or \cite{SimpBelRev}. 

The literature on simplicial semantics divides roughly into two approaches to defining valuations for propositional atoms. Loosely speaking, the first approach assigns truth values directly to the facets \cite{Death,FA,SimpSet},	while the second ``vertex-based'' approach assigns them (in a partial way) to the vertices and then ``lifts'' them to facets. \cite{DoA,KaSC,SimpDEL,HG} This paper will focus on the latter approach.

Let $\mathfrak{P}$ be a countable set of propositional atoms, and $Ag$ a finite set of agents. Let $N$ be a set of of \defin{nodes}, $V:\rightarrow Ag$ a function called the \defin{coloring} function, and $L:N\rightarrow 3^\mathfrak{P}$ a function which assigns each node to a set of literals, which we call the \defin{assignment}. The interpretation is that $L(n)(P)=1$ if and only if $P$ is associated with $n$, $L(n)(P)=0$ if and only if $\neg P$ is associated with $n$, and $L(n)(P)=2$ if and only if neither is associated with $n$. $L$ is the only difference between their presentation here as compared to \cite{BelSimp}. As mentioned above, $L$ assigns propositions to nodes.

Our first key idea is that we can use $N$, $V$, and $L$ to create a kind of \defin{maximal} simplicial complex. Like much of the previous literature, we will assume our simplicial complexes are uniquely colored. Specifically, each facet of our complexes has a dimension of size $|Ag|$, and no two nodes are associated to the same agent. Put formally, if $X\in S$ is such that for all $Y\in S$, if $X\subseteq Y$ then $X=Y$, we have that $|X|=|Ag|$, and, if $x,y\in X$, then if $V(x)=V(y)$, we have that $x=y$. We call this condition UCF for ``Uniquely Colored Facets''. The \defin{Maximal Complex} is the set of subsets of $N$ such that the associated logical content with that subset is consistent, and it satisfies the UCF condition. That is, it's the subsets $x\subseteq N$ such that $|x|=|A|$, The set of formulas assigned to $x$ by $L$ is consistent, and for all $u,v\in x$, $V(u)\neq V(v)$. We refer to $\mathfrak{M}((N,V,L)$ as the \defin{Maximal Complex} of $N$, $V$, and $L$. When the context is clear, we will refer to it simply as the maximal complex. The maximal complex can be defined set theoretically as follows:

\begin{align*}
	\mathfrak{M}(N,V,L):=\{y\in 2^N~|~\exists x\in 2^N (&y\subseteq x\\&\wedge\neg(\exists P\in\mathfrak{P}(\exists u,v\in x(L(u)(P)=1\wedge L(v)(P)=0))) \\&\wedge (|x|=|Ag|)
	\\&\wedge (\forall u,v\in x(V(u)\neq V(v))))\}
\end{align*}

Note that $\mathfrak{M}(N,V,L)$ is a simplicial complex. This is because, if $Y\in \mathfrak{M}(N,V,L)$, there is a set $X$ such that $Y\subseteq X$ and $X$ satisfies certain properties. If $Z\subseteq Y$, then the same set $X$ suffices to show that $Z\in \mathfrak{M}(N,V,L)$. All complexes we consider in this paper will be UCF subcomplexes of the maximal complex. 

If $S$ is a simplicial complex, let $\mathcal{F}(S)$ denote the facets of $S$. Since all simplicial complexes we consider are UCF, we can define projection functions $\pi_a:\mathcal{F}(S)\rightarrow N$ given by $\pi(a)(X)=V^{-1}(a)\cap X$. That is, $\pi_a$ maps each facet to the unique $a$-colored node it contains.

We can now define a \defin{Simplicial Model For Knowledge}:

\begin{definition}[Simplicial Model for Knowledge]
	A \defin{Simplicial Model for Knowledge} $\M$ is a tuple $(N,V,L,S)$ where $N$ is a nonempty set of nodes, $V$ is a coloring function, $L$ is an assignment function, and $S$ is a UCF subcomplex of $\mathfrak{M}(N,V,L)$.
\end{definition}

In this setting, the language $\L_{K}(Ag)$ recursively defined by
$$\varphi ::= P \, | \, \bot \, | \, \phi \lthen \psi \, | \, K_{a}\phi,$$
where $P \in \mathfrak{P}$ and $a \in Ag$, can be interpreted in simplicial models for knowledge as follows for any $X\in\F(S)$:
\begin{align*}
	\M,X & \models P \text{ iff } \exists a\in Ag(L(\pi_a(X))(P)=1)\\
	\M,X & \nvDash \bot\\
	\M,X & \models \phi \lthen \psi \text{ iff } \M,X \models \phi \text{ implies } \M,X \models \psi\\
	\M,X & \models K_a \phi \text{ iff } \forall Y \in \F(S) \text{ if } \pi_a(Y) = \pi_a(X) \text{ then } \M,Y \models \phi
\end{align*}

In order to incorporate belief revision, as we will do in Section \ref{3sec:Rev}, it will be necessary to talk about models for belief rather than models for knowledge. The following definition is borrowed from \cite{SimpBelRev}:

\begin{definition}[Simplicial Model for Belief]
	A \defin{Simplicial Model for Belief} $\M$ is a tuple $(N,V,L,S,\{S_a\}_{a\in Ag})$ where $N$ is a nonempty set of nodes, $V$ is a coloring function, $L$ is an assignment function, $S$ is a UCF subcomplex of $\mathfrak{M}(N,V,L)$, and for each $a\in Ag$, $S_a$ is a UCF subcomplex of $S$. Often, we will leave $S$ unspecified. In this case, we assume that $S=\mathfrak{M}(N,V,L)$.
\end{definition}

In this setting, the language $\L_{B}(Ag)$ recursively defined by
$$\varphi ::= P \, | \, \bot \, | \, \phi \lthen \psi \, | \, B_{a}\phi,$$
where $P \in \mathfrak{P}$ and $a \in Ag$, can be interpreted in simplicial models for knowledge as follows for any $X\in\F(S)$:
\begin{align*}
	\M,X & \models P \text{ iff } \exists a\in Ag(L(\pi_a(X))(P)=1)\\
	\M,X & \nvDash \bot\\
	\M,X & \models \phi \lthen \psi \text{ iff } \M,X \models \phi \text{ implies } \M,X \models \psi\\
	\M,X & \models B_a \phi \text{ iff } \forall Y \in \F(S_a) \text{ if } \pi_a(Y) = \pi_a(X) \text{ then } \M,Y \models \phi
\end{align*}

Following common convention, if $\M$ is a simplicial model for either knowledge or belief, the claim $x\in\M$ is short hand for $x\in N$ where $N$ is the set of nodes in $\M$.

Since belief complexes and how we draw them are relatively novel, we will look at a few examples. Consider Figure \ref{3fgr:SimpBelEx1}. In this example, there are three agents, $a$, $b$, and $c$. Hence, facets will be of size $3$ and are drawn as triangles. There are six nodes, $a_i$, $b_i$, and $c_i$ for $i\in 2$, and in each $S_i$ there are three facets, $\{a_1,b_1,c_1\}$, $\{a_1,b_1,c_0\}$, and $\{a_0,b_0,c_0\}$. There are three atomic propositions, $P_i$ for $i\in 2$, and $V(i_j)(P_k)=j$ if $k=i$ and $V(i_j)(P_k)=2$ otherwise.\footnote{It follows that this example can be given in the language of ``local variables''. \cite{KaSC} In fact, many of the examples and definitions in this paper are very easily lifted to that setting, as we are simply working in a slightly more general setting. See \cite{SimpBelRev} for more.}

\begin{figure}
	$$\begin{tikzcd}
		& \color{blue}b_0\color{black}(\neg P_b) \arrow[ld, no head] \arrow[dd, no head] & \color{red}a_1\color{black}(P_a) \arrow[rd, no head] \arrow[dd, no head]  &                                     \\
		\color{green}c_0\color{black}(\neg P_c) \arrow[rd, no head] \arrow[rru, no head] &                                                        &                                                         & \color{green}c_1\color{black}(P_c) \arrow[ld, no head] \\
		& \color{red}a_0\color{black}(\neg P_a)                                          & \color{blue}b_1\color{black}(P_b) \arrow[llu, no head]                                             
	\end{tikzcd}$$\caption{A simplicial model for belief with three agents and three worlds. In this example, $S_a=S_b=S_c$, and so we draw the complex as black rather than with distinct colors.}\label{3fgr:SimpBelEx1}
\end{figure}
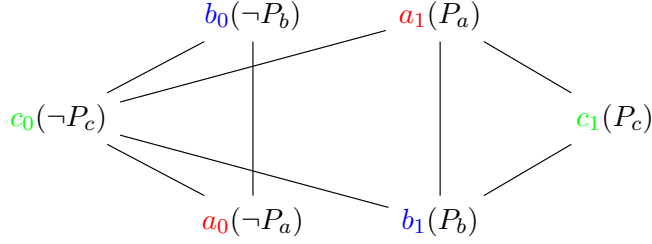

We should consider an example where the $S_i$ are distinct. Consider Figure \ref{3fgr:SimpBelEx2}. This example is drawn on the same set of nodes with the same valuation, but $S_a\neq S_b$ and $S_a\neq S_c$. Facets which are present in $S_a$ are drawn in red, those in $S_b$ are drawn in blue, and those in $S_c$ are drawn in green. If a facet is present in all three complexes, we draw it in black, as before. Hence, the two facets $\{a_1,b_1,c_1\}$ and $\{a_1,b_1,c_0\}$ are present in all three $S_i$. However, the facet $\{a_1,b_0,c_0\}$ is present only in $S_b$ and $S_c$.

\begin{figure}
	$$\begin{tikzcd}
		& \color{blue}b_0\color{black}(\neg P_b) \arrow[blue, ld, no head] \arrow[green, ld, no head, bend right] \arrow[blue, r, no head] \arrow[green, r, no head, bend left] & \color{red}a_1\color{black}(P_a) \arrow[rd, no head] \arrow[dd, no head]  &                                     \\
		\color{green}c_0\color{black}(\neg P_c) \arrow[rru, no head] &                                                        &                                                         & \color{green}c_1\color{black}(P_c) \arrow[ld, no head] \\
		& \color{red}a_0\color{black}(\neg P_a)                                          & \color{blue}b_1\color{black}(P_b) \arrow[llu, no head]                                             
	\end{tikzcd}$$\caption{A simplicial model for belief with three agents and three worlds. In this example, $S_a\neq S_b$, $S_a\neq S_c$, and $S_b=S_c$.}\label{3fgr:SimpBelEx2}
\end{figure}
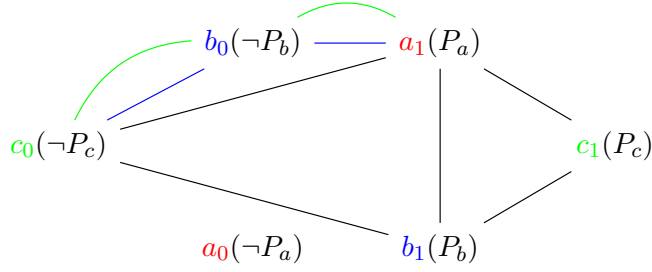

\section{Simplicial Complexes and Protocols}\label{3sec:SimpProt}

There is, at this point, a fairly well established discipline of applying combinatorial topology in modeling distributed protocols. \cite{Death,DoA,KaSC,SimpDEL,HG,FA,DCTCT} For those unfamiliar, we will briefly outline the broad methodology as sketched in the textbook ``Distributed Computing through Combinatorial Topology''. \cite{DCTCT} A ``protocol'' is a description of a sequence of signals between multiple processes, or agents. A good example is the ``Alternating Message Passing Protocol'' on pages 30-31 of the textbook. In this protocol, there are two agents, call them $a$ and $b$. Each agent has a privately held bit value, and is aware of their own bit value. However, each agent is uncertain of the other agent's bit value. We will denote these bit values as $P_a$ and $P_b$, respectively. So, $a$ is certain of the truth value of $P_a$, but uncertain of $P_b$, and vice versa. 

One agent is assumed to be the first sender, and we will assume this is $b$. In the protocol, the agents will send back and forth. So, once $a$ receives a message, $a$ becomes the next sender. However, the catch is that there's always a chance a message fails to be received. In this case, all sending stops.

One of the core ideas underlying applying simplicial complexes to these protocols is that we can describe each step of the protocol as a series of perspectives, and hence, each step as a simplicial model.\footnote{The textbook never actually interprets epistemic formulae in its models, but appeals to the intuition that two faces sharing the same perspective means that perspective is uncertain between those faces at many points throughout.} Consider the initial setup of the alternating message passing protocol. There are four persecptives, two for each agent. Let $a_1$ be the perspective where $P_a$ is true, and $a_0$ the perspective where it is false. Similarly with $b_1$ and $b_0$. This initial configuration is drawn on page 30 of ``Distributed Computing through Combinatorial Topology'', and we replicate that in Figure \ref{3Figure1}.

\begin{figure}
	$$\begin{tikzcd}
		{\color{red}a_0} \arrow[r, no head]  & {\color{blue}b_0} \arrow[d, no head]          \\
		{\color{blue}b_1} \arrow[u, no head] & {\color{red}a_1} \arrow[l, no head]
	\end{tikzcd}$$\caption{The Initial Configuration for the ``Alternating Message Passing Protocol''.}\label{3Figure1}
\end{figure}
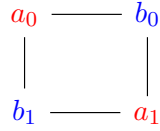

The main idea then is that perspectives will change and duplicate in each stage. Consider $a_0$. There are three possibilities for how the first step can go for this perspective. Either $P_b$ is received, $\neg P_b$ is received, or nothing is received. So, in the second step, $a_0$ has to split into three perspectives. A similar change of reasoning will show that $a_1$ also splits into three perspectives.

Now let's consider $b_0$. Really, only one thing can happen in the first stage, and that's that you send $\neg P_b$. Same for $b_1$. So, this leaves us with 3 copies of $a_0$, 3 copies of $a_1$, a single copy of $b_0$, and a single copy of $b_1$ in the second stage of the protocol. Again, this is drawn on page 30 of ``Distributed Computing through Combinatorial Topology'', and we replicate that drawing in Figure \ref{3Figure2}.

\begin{figure}
	$$
	\begin{tikzcd}
		&                                                                                                  & {\color{red}a_1}                                                           &                                                 \\
		& {\color{red}a_0} \arrow[r, no head]                          & {\color{blue}b_0} \arrow[d, no head] \arrow[r, no head] \arrow[u, no head] & {\color{red}a_0} \\
		{\color{red}a_1} \arrow[r, no head] & {\color{blue}b_1} \arrow[u, no head] \arrow[d, no head] & {\color{red}a_1} \arrow[l, no head]                                   &                                                 \\
		& {\color{red}a_0}                                                  &                                                                                                                &                                                
	\end{tikzcd}$$\caption{The second step for the ``Alternating Message Passing Protocol''.}\label{3Figure2}
\end{figure}
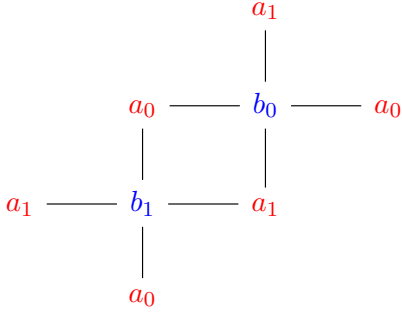

The same sort of reasoning, given the description of the signal pattern for the protocol, can give a very predictable way in which perspectives are duplicated from stage to stage of the protocol. The crucial idea is that we can describe a kind of morphism which maps protocols to their duplicates, or more generally, faces of the complex to faces composed of perspectives which are duplicates of those composing the initial face. This determines the ``pattern'' by which perspectives are duplicated. Ultimately, the textbook is interested in this sense of protocol. A protocol specifies how one complex transforms into another (how one stage proceeds to the next stage) by specifying a morphism which shows which faces map to which.

Let us spell this intuitive story out formally. The following is equivalent to Definition 8.4.1 in the textbook ``Distributed Computing through Combinatorial Topology'': \cite{DCTCT}

\begin{definition}[Simplicial Protocol]
	A \defin{Simplicial Protocol} is a triple $(\mathcal{I},\mathcal{P},\Xi)$ where $\mathcal{I}$ is a simplicial model for knowledge, called the \defin{``Input Complex''}, $\mathcal{P}$ is a simplicial model for knowledge called the \defin{``Protocol Complex''} (better thought of in our context as the output complex) and $\Xi:\mathcal{I}\rightarrow2^\mathcal{P}$ is called the \defin{``Execution Map''}, which takes faces in $\mathcal{I}$ to sets of faces in $\mathcal{P}$ such that if $\sigma\subseteq\tau$, $\Xi(\sigma)\subseteq\Xi(\tau)$, and $\Xi(\sigma\cap\sigma')=\Xi(\sigma)\cap\Xi(\sigma')$. Moreover, for any face $\sigma$ in $\mathcal{I}$, $\{a\in Ag~|~\exists x\in\sigma(V_\mathcal{I}(x)=a)\}=\{a\in Ag~|~\exists X\in\Xi(\sigma),x\in X(V_\mathcal{P}(x)=a)\}$. When this is true we say that $\Xi$ is \defin{chromatic}.\footnote{Intuitively, chromaticity simply means that if $x$ is a face whose colors are $C\subseteq Ag$, then $\Xi$ lifts $x$ to a set of faces each of whose colors are exactly $C$, in addition to all the subsets of those faces.} And, additionally, $\mathcal{P}=\bigcup_{\sigma\in\mathcal{I}}\Xi(\sigma)$.
\end{definition}

In general, $\Xi$ specifies which input perspectives go to which output perspectives. The restrictions on $\Xi$ bear further analysis. Consider in particular the intersection property: $\Xi(\sigma\cap\sigma')=\Xi(\sigma)\cap\Xi(\sigma')$. The textbook justifies it by saying no agent ``can ``tell'' whether the execution started with inputs from $\sigma$ or from $\sigma'$''. \cite{DCTCT} Here, ``tell'' refers to the usual understanding of knowledge in simplicial complexes, where the accessibility of two facets is given by their nonempty intersection.\footnote{This is one place where the textbook gestures at intuitions afforded by the simplicial semantics.} Suppose that a node $x$ is in the image of $\sigma$ and $\sigma'$. That is, $x\in\Xi(\sigma)\cap\Xi(\sigma')$. Then, it follows from our intersection rule that $x\in\Xi(\sigma\cap\sigma')$. By our intersection semantics in Section \ref{3sec:simpsem}, these are precisely the nodes which are in the image of those that cannot distinguish between $\sigma$ and $\sigma'$. So, nodes which cannot tell whether they are in the image of $\sigma$ or $\sigma'$ are precisely those whose preimage cannot tell whether they are in $\sigma$ or $\sigma'$.

We will quickly show how this definition applies to figure \ref{3Figure1} and figure \ref{3Figure2}. First, $\mathcal{I}$ is the complex in figure \ref{3Figure1}, and $\mathcal{P}$ is the complex in \ref{3Figure2}. The only faces in $\mathcal{I}$ are singletons and pairs, so for any singleton face of the form $\{a_i\}$ or $\{b_i\}$, we say that $\Xi$ maps $\{n_i\}$ to the set of all $\{n_i\}$ in $\mathcal{P}$. For any pair of the form $\{a_i,b_j\}$, we say that $\Xi$ maps $\{a_i,b_j\}$ to the set of all $\{a_i,b_j\}$ in $\mathcal{P}$, plus all of the singletons $\{a_i\}$ and $\{b_j\}$. It is easy to check that this satisfies all of the needed properties.

\section{Action Models and Protocols} \label{3sec:ActProt}

In this section, we will first briefly outline action models in their usual setting of Kripke or frame semantics for those unfamiliar with them. We will make use of definitions for action models over Kripke frames akin to those in \cite{action}. Then, we will discuss the existing work connecting action models to distributed protocols, before developing our own definition of actions similar to \cite{KaSC} and \cite{SimpDEL} in Section \ref{3sec:SimpAct}, though modified for our new particular simplicial semantics.

Consider the following informal description of a signal:

``$a$ sends a message to $b$ whose content is $P_a$. However, $c$ is unaware that any message was sent at all, and furthermore, there is a chance the message will not be received, in which case $b$ will think no message was sent.''\label{3Infsig}

What we would like to do is find a way of representing this signal such that, given an input model $\mathcal{M}$, define the model $\mathcal{M}'$ which represents this input after the above signal is sent. As it turns out, to model this signal as an action model, we need to represent it using something very similar to a Kripke frame. For those unfamiliar with the basics of Kripke semantics, we recommend \cite{DC5} or \cite{Bjorndahl2024}. For us, a \defin{Kripke Frame} is a tuple $\langle W,\{\mathcal{R}_a\}_{a\in Ag}\rangle$, and a \defin{Kripke model} is a tuple $\langle W,\{\mathcal{R}_a\}_{a\in Ag},\ell\rangle$ which augments a frame with a valuation function $\ell:\mathfrak{P}\rightarrow 2^W$. Since we want to build something akin to a Kripke frame, the first step is to describe the worlds which capture the above signal. As it stands, there are three. The first world, which we shall call $k$, is the world where a message is sent, but not received. The second world, which we shall call $m$, is the world where a message is sent, but is received. Naturally, $a$ cannot tell these worlds apart. However, $b$ can. At $m$, $b$ considers $m$ possible. However, at $k$, $b$ should not consider $k$ possible. After all, if the message is not received, $b$ thinks that no message was sent. This tells us that we need a third world, $l$, which is the world where nothing happens. At world $k$, $b$ falsely believes the world is $l$. Similarly, at both $k$ and $m$, $c$ falsely believes the world is $l$. The only curious touch is that, at world $l$, $a$ should think the world is $l$. After all, this is the world where no message was sent - if $a$ doesn't send a message at all, they think they are in world $l$. This description gives us the Kripke frame in Figure \ref{3fgr:KripFrameAct}

\begin{figure}
	$$
	\begin{tikzcd}
		& {l} \arrow[loop, distance=2em, in=125, out=55] &  \\
		{k} \arrow[green, ru, bend right] \arrow[blue, ru, bend left] \arrow[red, rr, leftrightarrow] \arrow[red, loop, distance=2em, in=305, out=235] &                                                      & {m} \arrow[green, lu, bend left] \arrow[red, loop, distance=2em, in=305, out=235] \arrow[blue, loop, distance=2em, in=125, out=55]
	\end{tikzcd}
	$$\caption{The Kripke Frame corresponding to the signal \ref{3Infsig}.}\label{3fgr:KripFrameAct}
\end{figure}
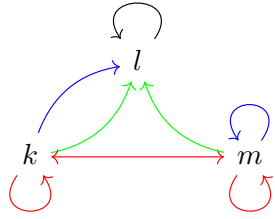

As before, red corresponds to $\mathcal{R}_a$, blue to $\mathcal{R}_b$, and green to $\mathcal{R}_c$. Black edges are those edges present in all three relations. However, the above frame is insufficient to describe the action. After all, these worlds are not compatible with every possible input. Imagine an input world $w$ where $\neg P_a$ was true. Presumably, since $P_a$ is a fact that $a$ is aware of, i.e., a secret pertaining to $a$ or a bit value which $a$ is aware of, it is not possible that $a$ sends out $P_a$ at $w$ (we are assuming that $a$ is not lying and the signal always goes through correctly if it goes through). Hence, $w$ is incompatible with $k$ and $m$. Actions use the notion of ``preconditions'' to bear this out. Basically, in order for world $k$ to obtain, it must be the case, and hence, it is a precondition, that $P_a$ is true. We will mark preconditions with round brackets, like we do for truth conditions in typical models, and let context differentiate. However, preconditions, at least for actions over Kripke frames, are distinct from truth conditions, as a world can have a precondition which is not an atomic formula. Writing in the preconditions gives us the model depicted in Figure \ref{3fgr:KripModAct}

\begin{figure}
	$$
	\begin{tikzcd}
		& {l(\top)} \arrow[loop, distance=2em, in=125, out=55] &                                                                                                                       \\
		{k(P_a)} \arrow[green, ru, bend right] \arrow[blue, ru, bend left] \arrow[red, rr, leftrightarrow] \arrow[red, loop, distance=2em, in=305, out=235] &                                                      & {m(P_a)} \arrow[green, lu, bend left] \arrow[red, loop, distance=2em, in=305, out=235] \arrow[blue, loop, distance=2em, in=125, out=55]
	\end{tikzcd}
	$$\caption{The Kripke Model corresponding to the signal \ref{3Infsig}.}\label{3fgr:KripModAct}
\end{figure}
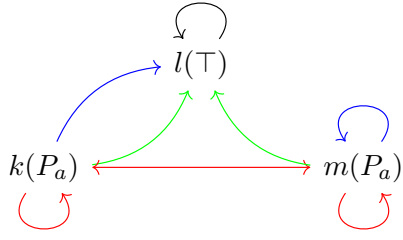

$l$ is compatible with any input, however, $k$ and $m$ require that $P_a$ to be true in order for them to obtain. Nicely, this is precisely the action given on page 22 of \cite{sus}, just with the precondition $K_a\varphi$ swapped out for $P_a$: when we turn to the simplicial setting in Section \ref{3sec:SimpAct}, these preconditions will be equivalent so long as we only ever assign $P_a$ to $a$-colored perspectives in the input, and hence, $P_a$ will be true if and only if $B_aP_a$ is true.

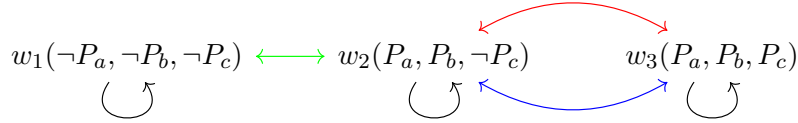
\begin{figure}
	$$
	\begin{tikzcd}
		{w_1(\neg P_a,\neg P_b,\neg P_c)} \arrow[green, r, leftrightarrow] \arrow[loop, distance=2em, in=305, out=235] & {w_2(P_a,P_b,\neg P_c)} \arrow[red, r, leftrightarrow, bend left] \arrow[blue, r, leftrightarrow, bend right] \arrow[loop, distance=2em, in=305, out=235] & {w_3(P_a,P_b,P_c)} \arrow[loop, distance=2em, in=305, out=235]
	\end{tikzcd}
	$$\caption{The Kripke Model we interpret as the input, or starting configuration.}\label{3fgr:KripModInput}
\end{figure}

Now that we understand what actions look like, let's imagine what updating an input with an action model might yield. Consider the input model as drawn in Figure \ref{3fgr:KripModInput}. This input is easy enough to interpret - it is the Kripke equivalent of Figure \ref{3fgr:SimpBelEx1}, using the translation given in \cite{SimpBelRev}. Following our heuristic description of how to update some input with an action, we know that worlds in the update are ordered pairs, whose first term is from the input, and whose second term is from the action. However, not all ordered pairs are allowed - if at a world $w$, the precondition for the action world $n$ is false, then the pair $(w,n)$ cannot be allowed in the action update. Worlds in the update model are typically given as ordered pairs, with a world from the input in the first place and a world from the output in the second place. The allowable pairs are those where the first world satisfies the preconditions of the second world. For our example, using figures \ref{3fgr:KripModAct} and \ref{3fgr:KripModInput}, this gives us the collection of worlds drawn in Figure \ref{3fgr:WorldOutput}.

\begin{figure}
	$$
	\begin{tikzcd}
		& {(w_2,m)} & {(w_3,m)} \\
		& {(w_2,k)} & {(w_3,k)} \\
		{(w_1,l)} & {(w_2,l)} & {(w_3,l)}
	\end{tikzcd}
	$$\caption{The worlds in the output where Figure \ref{3fgr:KripModInput} is the input and Figure \ref{3fgr:KripModAct} is the action.}\label{3fgr:WorldOutput}
\end{figure}

What should the truth conditions at these worlds be? Well, if $(w_2,k)$ is the world $w_2$, but where the message is sent but not received, then this should take on the same truth values as $w_2$. Indeed, actions don't change what's true at a world. Rather, they create copies of worlds which can be differently interpreted. This gives us the model as drawn in Figure \ref{3fgr:WorldValOutput}.

\begin{figure}
	$$
	\begin{tikzcd}
		& {(w_2,m)(P_a,P_b,\neg P_c)} & {(w_3,m)(P_a,P_b,P_c)} \\
		& {(w_2,k)(P_a,P_b,\neg P_c)} & {(w_3,k)(P_a,P_b,P_c)} \\
		{(w_1,l)(\neg P_a,\neg P_b,\neg P_c)} & {(w_2,l)(P_a,P_b,\neg P_c)} & {(w_3,l)(P_a,P_b,P_c)}
	\end{tikzcd}
	$$\caption{The worlds with truth values in the output where Figure \ref{3fgr:KripModInput} is the input and Figure \ref{3fgr:KripModAct} is the action.}\label{3fgr:WorldValOutput}
\end{figure}

All that is left to determine is each $\mathcal{R}_i$ for $i\in Ag$. When should $(w,n)\mathcal{R}_i(w',n')$? Well, $(w,n)$ is world $w$ where $n$ has occurred. So, if $w$ and $w'$ are differentiable, so should $(w,n)$ and $(w',n')$. The action does not change what is true at worlds, and so does not make worlds which were differentiable suddenly undifferentiable. However, if $n$ and $n'$ are differentiable, then $(w,n)$ and $(w',n')$ should be differentiable, even if $w$ and $w'$ are undifferentiable. After all, this is the main point - these were the same world, or at least undifferetiable worlds, but something different has occurred in each. This gives us a very natural definition: $(w,n)\mathcal{R}_i(w',n')$ if and only if $w\mathcal{R}_iw'$ and $n\mathcal{R}_in'$. This gives us the output model as drawn in Figure \ref{3fgr:KripModOutput}.

\begin{figure}
	$$
	\begin{tikzcd}
		& {(w_2,m)(P_a,P_b,\neg P_c)} \arrow[red, r, leftrightarrow, bend left] \arrow[blue, r, leftrightarrow, bend right] \arrow[green, dd, bend right] \arrow[blue, loop, distance=2em, in=125, out=55] \arrow[red, loop, distance=2em, in=215, out=145] \arrow[red, d, leftrightarrow] & {(w_3,m)(P_a,P_b,P_c)} \arrow[green, dd, bend left] \arrow[blue, loop, distance=2em, in=125, out=55] \arrow[red, loop, distance=2em, in=35, out=325] \arrow[red, d, leftrightarrow] \\
		& {(w_2,k)(P_a,P_b,\neg P_c)} \arrow[red, r, leftrightarrow, bend left] \arrow[green, d, bend right] \arrow[blue, d, bend left] \arrow[red, loop, distance=2em, in=215, out=145]                                           & {(w_3,k)(P_a,P_b,P_c)} \arrow[green, d, bend right] \arrow[blue, d, bend left] \arrow[red, loop, distance=2em, in=35, out=325]                                          \\
		{(w_1,l)(\neg P_a,\neg P_b,\neg P_c)} \arrow[green, r, leftrightarrow] \arrow[loop, distance=2em, in=305, out=235] & {(w_2,l)(P_a,P_b,\neg P_c)} \arrow[red, r, leftrightarrow, bend left] \arrow[blue, r, leftrightarrow, bend right] \arrow[loop, distance=2em, in=305, out=235]                         & {(w_3,l)(P_a,P_b,P_c)} \arrow[loop, distance=2em, in=305, out=235]
	\end{tikzcd}
	$$\caption{The Kripke model output where  Figure \ref{3fgr:KripModInput} is the input and Figure \ref{3fgr:KripModAct} is the action.}\label{3fgr:KripModOutput}
\end{figure}
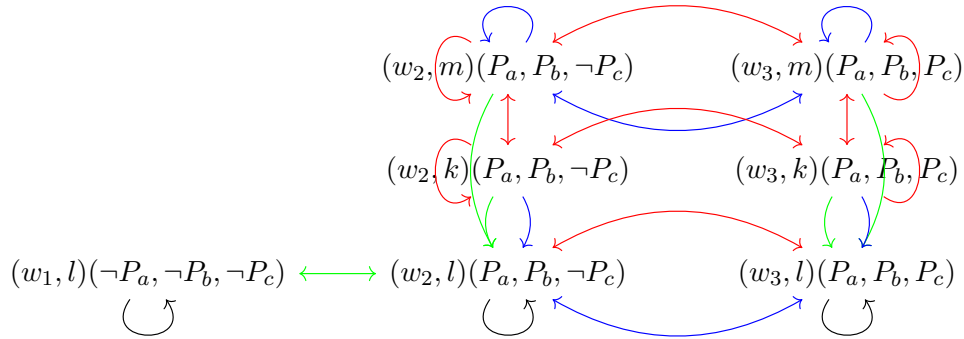

Using our heuristic understanding, it's not difficult to imagine what each world in the above model means. At any world $(w_j,n)$, $c$ should only consider possible those worlds $(w_h,l)$ such that $w_j\mathcal{R}_cw_h$, because $c$ always falsely believes no signal was sent. Similarly, at world $(w_2,k)$, all four of $(w_2,k)$, $(w_2,m)$, $(w_3,k)$, and $(w_3,m)$ are indistinguishable for $a$, as $a$ cannot tell apart $w_2$ and $w_3$, nor can $a$ tell apart $k$ from $m$. In general, one can think of an action update as creating (restricted) ``copies'' of the input model, and these ``copies'' are indistinguishable in terms of the relations in the action. Hence, the $l$ ``copy'' of the input model contains the only worlds which $c$ considers possible, etc.

This gives us the following formal definition of an action, equivalent to the one in \cite{action}:

\begin{definition}\label{3DefActKripke}
	
	Formally, an action model $A$ is a tuple $\langle W_A,\{\mathcal{R}_{A,a}\}_{a\in Ag},pre\rangle$\footnote{This notation is not standard in the literature, though it parallels the notation I will use later when developing simplicial actions.} such that $pre$ is a function from worlds $w\in W$ to sets of formulas.\footnote{$W_A$ is a set of worlds, and each $\mathcal{R}_{A,a}$ is a binary relation on those worlds.} Given a Kripke model $\mathcal{M}:=\langle W,\{\mathcal{R}_a\}_{a\in Ag},\ell\rangle$, and an action model $A$, we define the action update model $\mathcal{M}[A]:=\langle W[A],\{\mathcal{R}_a[A]\}_{a\in Ag},\ell[A]\rangle$ as follows:\footnote{$W[A]$ is a subset of $W\times W_A$, Each $\mathcal{R}_a[A]$ is a binary relation on $W[A]$, and $\ell[A]:\mathfrak{P}\rightarrow 2^{W[A]}$, just as in any Kripke model.}
	\begin{align*}
		W[A]&:=\{(x,y)\in W\times W_A|\mathcal{M},x\vDash pre(y)\}
		\\ \mathcal{R}_a[A](x,y)&:=\{(x',y')\in W[A]|x'\in \mathcal{R}_a(x)\wedge y'\in \mathcal{R}_{A,a}(y)\}
		\\ \ell[A](P)&:=\{(x,y)\in W[A]|x\in \ell(P)\}
	\end{align*}
	
\end{definition}

We conclude this section by outlining the existing applications of action models to distributed protocols. In general, the procedures for using action models to describe protocols are understood but new, showing up in \cite{SimpDEL}, \cite{KaSC}, \cite{KaG}, \cite{ActionDC}, \cite{CPM}, and \cite{CPaAM}. In Section \ref{3sec:Ex}, we will translate two of the examples of protocols in \cite{DCTCT} into action models, broadly following the trends of these papers. The key idea is that we can think of each ``step'' in the protocol as the application of an action. This is why we used the intuition of a ``signal'' when giving our account of actions. Not all actions can or should be interpreted as signals between agents (they are very general objects), but those actions that can be interpreted as signals can be interpreted as steps in a protocol. Suppose we are in a setting where agents have a privately held bit value. If the first step of a protocol in this setting is that all agents broadcast their bit value, with a chance that this broadcast is not received by some or all of the other agents, we can model this step as an action in the following sense: if $\M$ represents the protocol before this step, and $A$ is the action corresponding with this signal, $\M[A]$ is the protocol after this step. Hence, $\M[A]$ will be the input for the next step, which itself will be represented by a possibly distinct action, and so on.

Let's again consider the ``Alternating Message Passing Model''. \cite{DCTCT} All Kripke models we will consider in this example are equivalence relations. We again consider the case with two agents, call them $a$ and $b$. Each agent has a privately held bit value, which we will denote using the proposition $P_a$ and $P_b$, respectively. Hence, in the input, there are four worlds, corresponding with all possible configurations of the bit values. This is drawn in Figure \ref{3fgr:KAMPworlds}. It is easy to determine which of these worlds the agents cannot tell apart. Agent $a$ is uncertain about the truth value of $P_b$, but certain about the truth value $P_a$, and vice versa for agent $b$. So, agent $a$ cannot tell apart $w_0$ and $w_2$, and cannot tell apart $w_1$ and $w_3$. Hence, $w_0\mathcal{R}_aw_2$ and $w_1\mathcal{R}_aw_3$. We will denote $\mathcal{R}_a$ with red edges and $\mathcal{R}_b$ with blue edges. This is drawn in Figure \ref{3fgr:KAMPinput}. Call this model $I$.

\begin{figure}
	$$
	\begin{tikzcd}
		{w_0(P_a,P_b)}      & {w_1(\neg P_a,P_b)}      \\
		{w_2(P_a,\neg P_b)} & {w_3(\neg P_a,\neg P_b)}
	\end{tikzcd}
	$$\caption{The four worlds in the input for the ``Alternating Message Passing Protocol''.}\label{3fgr:KAMPworlds}
\end{figure}

\begin{figure}
	$$
	\begin{tikzcd}
		{w_0(P_a,P_b)} \arrow[blue, r, leftrightarrow]      & {w_1(\neg P_a,P_b)} \arrow[red, d, leftrightarrow]      \\
		{w_2(P_a,\neg P_b)} \arrow[red, u, leftrightarrow] & {w_3(\neg P_a,\neg P_b)} \arrow[blue, l, leftrightarrow]
	\end{tikzcd}
	$$\caption{The first stage of the ``Alternating Message Passing Protocol'' with two agents.}\label{3fgr:KAMPinput}
\end{figure}

Recall how the ``Alternating Message Passing Protocol'' proceeds. In this protocol, one agent is randomly selected as the first sender. For our purposes, this will be agent $a$. In the first step, agent $a$ sends her bit value to agent $b$. However, there is a chance that this message is not received. This gives us four worlds. There are two worlds where $P_a$ is sent, and two worlds where $\neg P_a$ is sent. In one of the worlds where $P_a$ is sent, it is received by agent $b$, and in another, it is not. Similarly for the worlds where $\neg P_a$ is sent. Let $1_r$ denote the world where $P_a$ is sent and received, $1_{\sim r}$ denote the world where $P_a$ is sent and not received, $0_{r}$ the world where $\neg P_a$ is sent and received, and $0_{\sim r}$ the world where $\neg P_a$ is sent but not received. Hence, $1_r\mathcal{R}_a1_{\sim r}$ and  $0_r\mathcal{R}_a0_{\sim r}$. If $b$ is in $1_r$ they know they have received the message $P_a$, and hence only consider $1_r$ possible. Similarly for $0_r$. However, $b$ cannot tell apart $1_{\sim r}$ and $0_{\sim r}$, as these are the worlds where no message is received. Note that $a$ will only send their actual bit value. Hence, $a$ will only send $P_a$ if $P_a$ is true. So, the precondition for $1_r$ and $1_{\sim r}$ is $P_a$, and the precondition for $0_r$ and $0_{\sim r}$ is $\neg P_a$. We draw this in Figure \ref{3fgr:KAMPaction}. Call this action $A_a$. Hence, we can compute $I[A_a]$. This is drawn in Figure \ref{3fgr:KAMPoutput}. It's easy to check that $I[A_a]$ corresponds with the ``One Step'' drawing on page 30 in \cite{DCTCT}, using a translation between simplicial models and Kripke models similar to those given in \cite{KaSC}.\footnote{Despite being called a ``graph'', the drawings in Section 2 of \cite{DCTCT} should be thought of as simplicial models for knowledge with two agents.} 

After $b$ receives a message from $a$, she sends a response which consists of her own bit value. As before, there is a chance this message is not received. This action is the exact same as $A_a$, just with the agents swapped. We call this action $A_b$ and draw it in Figure \ref{3fgr:KAMPactionb}. Unfortunately, $(I[I_a])[I_b]$ does not correspond to the ``Two Steps'' drawing in \cite{DCTCT}. The reason is because even worlds where $a$'s message failed to be received may become worlds where $b$ still sends a message as we have defined things. In principle, however, accounting for this would be easy. One simply must define ``postconditions'' where what is true at a world is ``overwritten'', as is done for the action models in \cite{KaSC}. This would let us apply a tag to worlds like $(w_0,1_{\sim r})$ which reads ``not received'', and we can take as a precondition for all worlds in $I_b$ that this tag be false. However, rather than explore this more carefully here, we will leave a formal exploration of postconditions for Section \ref{3sec:SimpAct}. In Section \ref{3sec:Ex}, we will use postconditions to fully capture the Alternating Message Passing protocol.

\begin{figure}
	$$
	\begin{tikzcd}
		1_r(P_a) \arrow[red, r, leftrightarrow] & 1_{\sim r}(P_a) \arrow[blue, r, leftrightarrow] & 0_{\sim r}(\neg P_a) \arrow[red, r, leftrightarrow] & 0_r(\neg P_a)
	\end{tikzcd}
	$$\caption{The action representing the first step of the ``Alternating Message Passing Protocol'': Agent $a$ sends her true bit value to agent $b$ with a chance the message fails to be received.}\label{3fgr:KAMPaction}
\end{figure}

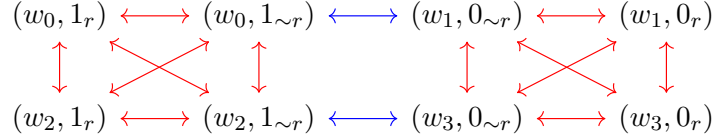
\begin{figure}
	$$
	\begin{tikzcd}
		{(w_0,1_r)} \arrow[red, d, leftrightarrow] \arrow[red, rd, leftrightarrow] & {(w_0,1_{\sim r})} \arrow[blue, r, leftrightarrow] \arrow[red, l, leftrightarrow] \arrow[red, ld, leftrightarrow] & {(w_1,0_{\sim r})} \arrow[red, d, leftrightarrow] \arrow[red, rd, leftrightarrow] & {(w_1,0_r)} \arrow[red, l, leftrightarrow] \arrow[red, ld, leftrightarrow] \\
		{(w_2,1_r)} \arrow[red, r, leftrightarrow]            & {(w_2,1_{\sim r})} \arrow[red, u, leftrightarrow]                      & {(w_3,0_{\sim r})} \arrow[blue, l, leftrightarrow] \arrow[red, r, leftrightarrow]  & {(w_3,0_r)} \arrow[red, u, leftrightarrow]           
	\end{tikzcd}
	$$
	\caption{The output of applying the action shown in Figure \ref{3fgr:KAMPaction} to the input shown in Figure \ref{3fgr:KAMPinput}.} \label{3fgr:KAMPoutput}
\end{figure}

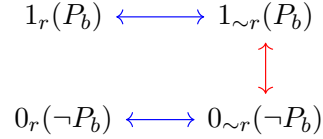
\begin{figure}
	$$
	\begin{tikzcd}
		1_r(P_b) \arrow[blue, r, leftrightarrow] & 1_{\sim r}(P_b) \arrow[red, d, leftrightarrow] \\
		0_r(\neg P_b)           & 0_{\sim r}(\neg P_b) \arrow[blue, l, leftrightarrow]
	\end{tikzcd}
	$$
	\caption{The action representing the first step of the signal describes as follows: agent $b$ sends her true bit value to agent $a$ with a chance the message fails to be received.} \label{3fgr:KAMPactionb}
\end{figure}

Our next goal will be to apply the notion of action models to simplicial semantics, similar to \cite{Death} and \cite{KaSC}, though slightly modified to apply to our particular simplicial semantics. After this, we will modify this definition of action to account for the notion of belief revision given in \cite{SimpBelRev}. Modifying action models where agents learn via belief revision is also an explored area, with work such as \cite{BRDC1} and \cite{BRDC2}. However, none of these models will use the specific model of revision we will motivate in Section \ref{3sec:SimpAct}. That is because this work on using revision in protocols does not use the formal apparati for describing protocols in simplicial complexes as seen in \cite{DCTCT}, and thus cannot take advantage of the model of revision presented in \cite{SimpBelRev}.

\section{Simplicial Actions} \label{3sec:SimpAct}

Before we define action models in the simplicial setting, we need to give three important definitions. First, $\pi_0$ and $\pi_1$ are the projections onto the first and second coordinates of an ordered pair, respectively. That is, for any ordered pair $(x,y)$, $\pi_0((x,y))=x$ and $\pi_1((x,y))=y$. For the third definition, let $N$ be a set of nodes, $V$ a coloring function, and $L$ an assignment. Suppose $F$ is a face in $\mathfrak{M}(N,V,L)$. We say that $F$ is \defin{consistent with respect to $L$} when for all $x,y\in F$, there is no $P\in\mathfrak{P}$ such that $L(x)(P)=1$ and $L(y)(P)=0$.

We can now give the definitions for simplicial actions formally. In general, we will use definitions for actions similar to those in \cite{Death}, \cite{KaSC}, and \cite{SimpDEL}. In particular, as in \cite{KaSC}, actions will have both pre and post conditions. Importantly, we will assume that both preconditions and postconditions are atomic formulae in this paper. It is essential that they be non-modal formulae, unlike in the general case of action models, so that we can interpret revision, and restricting to the atomic case is sufficient for our examples. Expanding preconditions and postconditions beyond atomic formulae is an area for future research. In drawings, we will denote preconditions with round brackets, and postconditions with square brackets. 

\begin{definition}[Simplicial Action]
	A \defin{Simplicial Action} $A$ is a tuple $\langle N_A,V_A,L_A,S_A,Post\rangle$, where $Post:N_a\rightarrow 3^\mathfrak{P}$.\footnote{$N_A$ is a set of nodes, $V_A:N_A\rightarrow Ag$, $L_A:N_A\rightarrow 3^\mathfrak{P}$, and $S_A$ is a UCF subcomplex of $\mathfrak{M}(N_A,V_A,L_A)$, just as in a regular simplicial model.}
	
	Given a simplicial model for knowledge $\mathcal{M}=\langle N,V,L,S\rangle$, and a simplicial action $A=\langle N_A,V_A,L_A,S_A,Post\rangle$, we define the update model $\mathcal{M}[A]:=\langle N[A],V[A],L[A],S[A]\rangle$\footnote{$N[A]$ is a subset of $N\times N_A$, $V[A]:N[A]\rightarrow Ag$, $L[A]:N[A]\rightarrow 3^\mathfrak{P}$, and $S[A]$ is a UCF subcomplex of $\mathfrak{M}(N[A],V[A],L[A])$, just as in a regular simplicial model.} as follows:
	
	\begin{align*}
		N[A]&:=\{(x,y)\in N\times N_A|V(x)=V_A(y)\\
		&\wedge\forall P\in\mathfrak{P}((L(x)(P)=1\rightarrow L_A(y)(P)\neq 0)\wedge(L(x)(P)=0\rightarrow L_A(y)(P)\neq 1))\}
		\\ V[A]((x,y))&:=V(x)
		\\ L[A]((x,y))(P)&:=\begin{cases}
			& Post(y)(P)\;\text{if}\;Post(y)(P)\neq 2 \\
			& L_A(y)(P)\;\text{if}\;L(x)(P)=2 \\
			&L(x)\;\text{otherwise}
		\end{cases} 
		\\ \mathcal{F}(S[A])&:=\left \{F\in 2^{N[A]}~\middle |~\begin{matrix}
			& \pi_0(F)\in\mathcal{F}(S) \\
			& \wedge\pi_1(F)\in\mathcal{F}(S_A) \\
			& \wedge F\text{ is consistent w.r.t. }L[A] \\
		\end{matrix} \right \}
	\end{align*}
\end{definition}

Note that if $F\in 2^{N[A]}$, $\pi_0(F)=\{x~|~\exists y\in S_A((x,y))\in F)\}$, and similarly for $\pi_1$. Let us show that for any simplicial model for knowledge $\M$ and simplicial action $A$, $\M[A]$ satisfies the UCF property.

\begin{theorem}\label{3thm:simpUCF} Let $\M$ be a simplicial model for knowledge and $A$ a simplicial action model. Then the output model $\mathcal{M}[A]$ also satisfies the UCF property.
	\begin{proof}
		Section \ref{3prf:simpUCF}
	\end{proof}
\end{theorem}

Our definition of $N[A]$ is not the only one which could be motivated. In fact, all three of the definitions below are defensible.

\begin{enumerate}
	\item $\forall P\in\mathfrak{P}((L(x)(P)=1\rightarrow L_A(y)(P)=1)\wedge(L(x)(P)=0\rightarrow L_A(y)(P)=0))$
	\item $\forall P\in\mathfrak{P}((L_A(y)(P)=1\rightarrow L(x)(P)=1)\wedge(L_A(y)(P)=0\rightarrow L_A(x)(P)=0))$
	\item $\forall P\in\mathfrak{P}((L(x)(P)=1\rightarrow L_A(y)(P)\neq 0)\wedge(L(x)(P)=0\rightarrow L_A(y)(P)\neq 1))$
\end{enumerate}

We will refer to these as modes 1-3. Indeed, mode 3 is the one introduced in our definition above. For our purposes in this paper, it will suffice.

We now show the sense in which action models are protocols. Let $I$ be a simplicial model for knowledge and $A$ an action model. Intuitively, our protocol complex should be given by $I[A]$. What we need to define is an execution map $\Xi:I\rightarrow 2^{I[A]}$. As it turns out, the correct map is given by 

$$\Xi(F):=\{X\in I[A]~|~\pi_0(X)\subseteq F\}$$

Where $X\in I[A]$ is shorthand for ``$X$ is a face in the simplicial complex of $I[A]$. This is analogous to the $\pi_\mathcal{I}$ function from \cite{SimpDEL}. This gives us the following:

\begin{theorem}\label{3thm:simpprot} Let $I$ be a simplicial model for knowledge and $A$ an action model. The tuple $(I,I[A],\Xi)$ is a Simplicial Protocol.
	\begin{proof}
		Section \ref{3prf:simpprot}
	\end{proof}
\end{theorem}

\subsection{Examples} \label{3sec:Ex}

As we have already mentioned, there is existing work which establishes the connection between action models and protocols. \cite{SimpDEL,KaSC,KaG,ActionDC,CPM,CPaAM} In particular, \cite{KaSC} and \cite{SimpDEL} even give some examples. We will do much the same, using examples from the text ``Distributed Computing through Combinatorial Topology'' as our baseline. \cite{DCTCT} These examples are nice because they are already given in the framework of simplicial complexes. All we need to do is prove that it is possible to represent these existing protocols using action models.

We can now fully capture the Alternating Message Passing Protocol on page 30 of \cite{DCTCT}. In this protocol, there are two agents, call them $a$ and $b$. As usual, these will be the red and blue agents respectively. The initial arrangement is that each agent has a personal bit value, and they will be sharing that bit value with each other. As above, we will call these bit values $P_a$ for $a$'s bit value and $P_b$ for $b$'s bit value. One agent is designated the first sender, and we shall pick $b$. The input complex will have four perspectives, two for each agent corresponding to the possible bit values, and total uncertainty. Since $b$ is the first sender, the two $b$ perspectives will be marked with the literals $S$ and $\neg R$, and since $a$ is the first received, they will be marked with the literals $\neg S$ and $R$. This gives us the input as drawn in Figure \ref{3Figure26}. Note how this is the same as the first stage drawing in \cite{DCTCT}, which we replicated in Figure \ref{3Figure1}.

\begin{figure}
	$$\begin{tikzcd}
		{\color{red}a_1\color{black}(P_a,\neg S,R)} \arrow[r, no head]  & {\color{blue}b_1\color{black}(P_b,S,\neg R)} \arrow[d, no head]          \\
		{\color{blue}b_0\color{black}(\neg P_b,S,\neg R)} \arrow[u, no head] & {\color{red}a_0\color{black}(\neg P_a,\neg S,R)} \arrow[l, no head]
	\end{tikzcd}$$\caption{The Initial Configuration for the ``Alternating Message Passing Protocol''.}\label{3Figure26}
\end{figure}
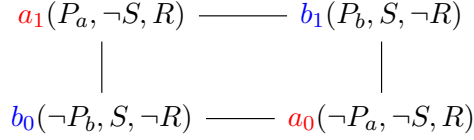

Now we must describe the protocol. The main idea of the alternating message protocol is that agents are sending messages back and forth, and there is always a chance a message does not go through. If an agent doesn't receive a message, they don't send. Let's imagine what happens when $b$ sends the first message. There are two $b$ perspectives, corresponding to where $b$ sends $P_b$ and where $b$ sends $\neg P_b$. These should have preconditions $S$, $\neg R$, and $P_b$ or $\neg P_b$ as appropriate. We will call these $1$ and $2$. By contrast, there are three $a$ perspectives. These correspond with $a$ learning $P_b$, $a$ learning $\neg P_b$, and $a$ learning nothing. We will call these $3$, $5$, and $4$ respectively. The preconditions on $3$ and $5$ are $\neg S$ and $R$, while the precondition on $4$ is merely $\neg S$. This last point will be important for repeated applications of the protocol - any $a$ perspective which is not a sender can receive nothing, possibly after having received nothing for many turns. This partially defined action model is drawn in Figure \ref{3Figure27}.

\begin{figure}
	$$
	\begin{tikzcd}
		{\color{blue}1\color{black}(P_b,S,\neg R)} \arrow[d, no head] &                                                                                                                         & {\color{blue}2\color{black}(\neg P_b,S,\neg R)} \arrow[d, no head] \\
		{\color{red}3\color{black}(\neg S,R)}                         & {\color{red}4\color{black}(\neg S)} \arrow[lu, no head] \arrow[ru, no head] & {\color{red}5\color{black}(\neg S,R)}
	\end{tikzcd}$$\caption{The fragment of the action for the ``Alternating Message Passing Protocol'' corresponding with $b$ sharing their bit value, sans postconditions.}\label{3Figure27}
\end{figure}
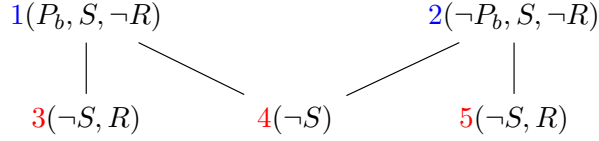

We now must draw the postconditions. Consider perspective $1$. Having just sent, they turn into a receiver. So the postconditions here should be $\neg S$ and $R$. Same for $2$. Similarly, receivers become senders, so the postconditions on $3$ and $5$ should be $S$ and $\neg R$. Lastly, the postcondition on $4$ should be $\neg R$. If $a$ receives nothing, they continue to not be a sender, but now are also not a receiver. This action model is drawn in Figure \ref{3Figure28}.

\begin{figure}
	$$
	\begin{tikzcd}
		{\color{blue}1\color{black}(P_b,S,\neg R)[\neg S,R]} \arrow[d, no head] &                                                                                                                         & {\color{blue}2\color{black}(\neg P_b,S,\neg R)[\neg S,R]} \arrow[d, no head] \\
		{\color{red}3\color{black}(\neg S,R)[S,\neg R]}                         & {\color{red}4\color{black}(\neg S)[\neg R]} \arrow[lu, no head] \arrow[ru, no head] & {\color{red}5\color{black}(\neg S,R)[S,\neg R]}                            \end{tikzcd}$$\caption{The fragment of the action for the ``Alternating Message Passing Protocol'' corresponding with $b$ sharing her bit value.}\label{3Figure28}
\end{figure}
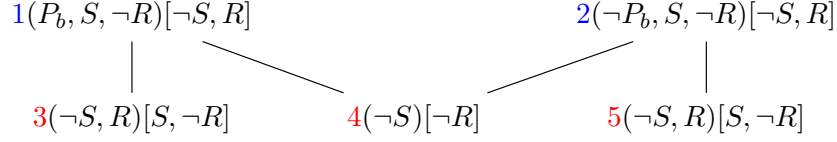

The last step in the protocol is to create the mirror image of this for when $a$ is the sender and attach it to the existing drawing. This is seen in Figure \ref{3Figure29}.

\begin{figure}
	$$
	\begin{tikzcd}
		{\color{blue}6\color{black}(\neg S,R)[S,\neg R]}                        & {\color{blue}7\color{black}(\neg S)[\neg R]} \arrow[ld, no head] \arrow[rd, no head]                   & {\color{blue}8\color{black}(\neg S,R)[S,\neg R]}                             \\
		{\color{red}9\color{black}(P_a,S,\neg R)[\neg S,R]} \arrow[u, no head]  &                                                                                                                         & {\color{red}10\color{black}(\neg P_a,S,\neg R)[\neg S,R]} \arrow[u, no head] 
	\end{tikzcd}$$\caption{The fragment of the action for the ``Alternating Message Passing Protocol'' corresponding with $a$ sharing her bit value.}\label{3Figure29}
\end{figure}
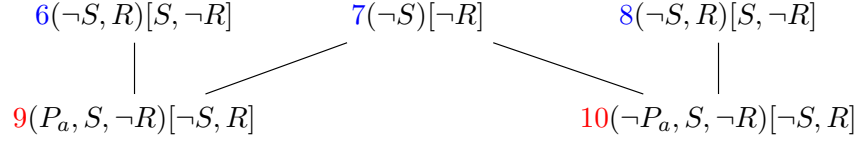

The last thing to acknowledge is that an edge must be drawn between $4$ and $7$ - faces between two edges where nothing is received continue to be faces after any step of the protocol, as these are the worlds where nothing happens. This gives us the complete action model, which is drawn in Figure \ref{3Figure30}.

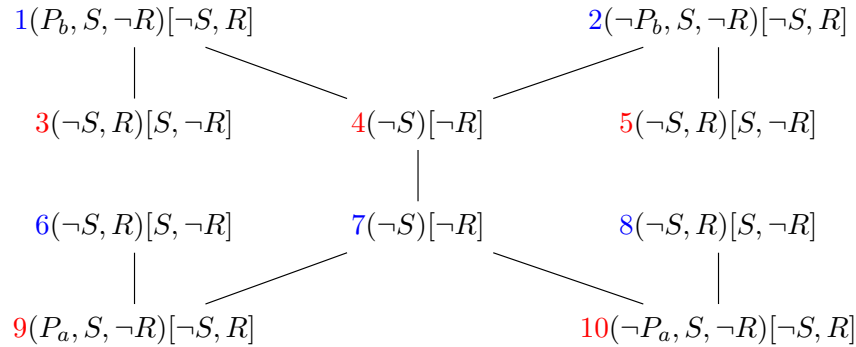
\begin{figure}
	$$
	\begin{tikzcd}
		{\color{blue}1\color{black}(P_b,S,\neg R)[\neg S,R]} \arrow[d, no head] &                                                                                                                         & {\color{blue}2\color{black}(\neg P_b,S,\neg R)[\neg S,R]} \arrow[d, no head] \\
		{\color{red}3\color{black}(\neg S,R)[S,\neg R]}                         & {\color{red}4\color{black}(\neg S)[\neg R]} \arrow[d, no head] \arrow[lu, no head] \arrow[ru, no head] & {\color{red}5\color{black}(\neg S,R)[S,\neg R]}                              \\
		{\color{blue}6\color{black}(\neg S,R)[S,\neg R]}                        & {\color{blue}7\color{black}(\neg S)[\neg R]} \arrow[ld, no head] \arrow[rd, no head]                   & {\color{blue}8\color{black}(\neg S,R)[S,\neg R]}                             \\
		{\color{red}9\color{black}(P_a,S,\neg R)[\neg S,R]} \arrow[u, no head]  &                                                                                                                         & {\color{red}10\color{black}(\neg P_a,S,\neg R)[\neg S,R]} \arrow[u, no head] 
	\end{tikzcd}$$\caption{The action for the ``Alternating Message Passing Protocol''.}\label{3Figure30}
\end{figure}

Now let's compute the output of applying this action model to the specified input. The output vertices are $(a_1,4)$, $(a_1,3)$, $(a_1,5)$, $(b_1,1)$, $(b_0,2)$, $(a_0,4)$, $(a_0,3)$, $(a_0,5)$. This gives the output as drawn in Figure \ref{3Figure31}. Note how this is the same as the drawing of the second step for the protocol as given in \cite{DCTCT}, which we replicated in Figure \ref{3Figure2}. There are new perspectives created for when $a$ learns the truth value of $y$, but also the original $a$ perspectives are present, now corresponding to a failed signal. Let's apply the action again. The output vertices are $((a_1,4),4)$, $((a_0,4),4)$, $((b_1,1),7)$, $((b_1,1),6)$, $((b_1,1),8)$, $((b_0,2),7)$, $((b_0,2),6)$, $((b_0,2),8)$, $((a_0,3),10)$, $((a_1,3),9)$, $((a_0,5),10)$, $((a_1,5),9)$. This again gives the result seen on page 30 of \cite{DCTCT}.

\begin{figure}
	$$
	\begin{tikzcd}
		&                                                                                                  & {(\color{red}a_0,3\color{black})(\neg P_a,S,\neg R)}                                                           &                                                 \\
		& {(\color{red}a_1,4\color{black})(P_a,\neg S,\neg R)} \arrow[r, no head]                          & {(\color{blue}b_1,1\color{black})(P_b,\neg S,R)} \arrow[d, no head] \arrow[r, no head] \arrow[u, no head] & {(\color{red}a_1,3\color{black})(P_a,S,\neg R)} \\
		{(\color{red}a_0,5\color{black})(\neg P_a,S,\neg R)} \arrow[r, no head] & {(\color{blue}b_0,2\color{black})(\neg P_b,\neg S,R)} \arrow[u, no head] \arrow[d, no head] & {(\color{red}a_0,4\color{black})(\neg P_a,\neg S,\neg R)} \arrow[l, no head]                                   &                                                 \\
		& {(\color{red}a_1,5\color{black})(P_a,S,\neg R)}                                                  &                                                                                                                &                                                
	\end{tikzcd}$$\caption{The output of applying the action in Figure \ref{3Figure30} to the input in Figure \ref{3Figure26}. }\label{3Figure31}
\end{figure}
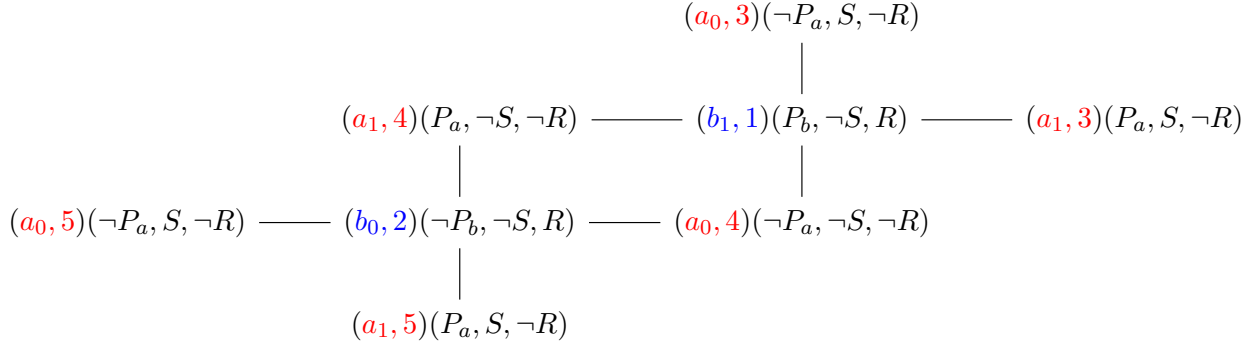

An important observation - when creating the action, we have to think through copies of perspectives. Even though creating a ``base'' copy of the $b_1$ world in the input (i.e., pairing it with $7$) would seem to make sense, it doesn't. We need to transform this perspective from a sender to a non-sender. This is the role $1$ and $2$ play in the action. The base ($4$ and $7$) duplicate everything up to the sender perspective. Then, $1$ duplicates the sender, while $3$ duplicates the receiver (already duplicated by $4$, hence an additional copy).

Let us prove by induction that this example works no matter how many times we apply the protocol. The above examples are our base case. Suppose Output(n) is a square of four nodes, with two strings off each blue node, alternating in color, of length n (so four strings total). Everything is $\neg S,\neg R$, except the last two nodes on each string, which are $\neg S, R$ just before the tip and $S,\neg R$ at the tip. We will show that applying the action increases the length of these strings by 1. Suppose, without loss of generality, that the tips of the strings are red. Then $4$ and $7$ copy everything up until the tip of each string. Each tip is copied by either $9$ or $10$, and the nodes just before the tip are additionally copied by $6$ and $8$. Because the tip of each node is either copied by $9$ or $10$, the additional copies from $6$ uniquely attach to their respective $9$-tipped string, and the $8$ copies their respective $10$-tipped string, as desired.

We will now turn towards our next example, the ``Layered Message Passing Protocol''.  Here, the two agents simultaneously write their bit value to a ledger at each time interval (so the protocol is synchronous). However, at each interval, there is a chance that at most one agent will fail. Upon observing a failure, the agents will stop. We will consider the restricted situation where $a$ has bit value $0$ and $b$ has bit value $1$ first. Each of course starts as a sender, giving us the input as drawn in Figure \ref{3Figure32}.

\begin{figure}
	$$\begin{tikzcd}
		\color{red}a\color{black}(\neg P_a,S) \arrow[r, no head] & \color{blue}b\color{black}(P_b,S)
	\end{tikzcd}$$\caption{Restricted input for the ``Layered Message Passing Protocol'' where $a$ has bit value $0$ and $b$ has bit value $1$.}\label{3Figure32}
\end{figure}

The protocol itself is trickier. The action has four fairly obvious perspectives, these being where $a$ receives the message, $a$ does not, and the same for $b$. But there are two other perspectives, one for each agent, which are less obvious because they are not relevant in the first step of the protocol. That is the perspective of $a$ where at some point in the past, they failed to receive a message, and therefore on this particular step of the protocol did not send any message. All told, this gives the action model drawn in Figure \ref{3Figure33}.  For example, perspective $2$ is one where $P_b$ sent, but $b$ failed to receive. $4$ is where $P_b$ sent and received, and $6$ is where $b$ did not send. Same for the $a$ perspectives. The most nonideal thing about this setup is the need for the $\neg P_a$ pre and post-condition in $4$, and the similar setup in $3$. These will be necessary to prevent unwanted duplications, corresponding to receiving $\neg P_a$ and then $P_a$, in the setup where the bit values are not fixed in the input, as we shall see. But as presented, they do seem overly restricted. Other workarounds would be ideal, though non are forthcoming. This too would be a useful avenue for future research.

\begin{figure}
	$$\begin{tikzcd}
		{\color{blue}2\color{black}(P_b,S)[\neg S]} \arrow[r, no head] \arrow[rd, no head] & {\color{red}3\color{black}(\neg P_a,P_b,S)[P_b]} \arrow[r, no head]           & {\color{blue}4\color{black}(\neg P_a,P_b,S)[\neg P_a]} \arrow[r, no head] & {\color{red}5\color{black}(\neg P_a,S)[\neg S]} \arrow[ld, no head] \\
		& {\color{red}1\color{black}(\neg P_a,\neg S)} \arrow[r, no head] & {\color{blue}6\color{black}(P_b,\neg S)}          &                                                               
	\end{tikzcd}$$\caption{The action for the restricted case of the ``Layered Message Passing Protocol'' where $a$ has bit value $0$ and $b$ has bit value $1$.}\label{3Figure33}
\end{figure}
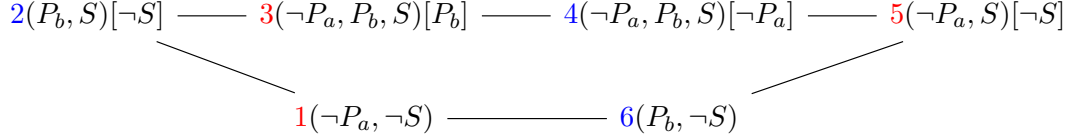

\begin{figure}
	$$\begin{tikzcd}
		{\color{red}(a,5)\color{black}(\neg P_a,\neg S)} \arrow[r, no head] & {\color{blue}(b,4)\color{black}(\neg P_a,P_b,S)} \arrow[r, no head] & {\color{red}(a,3)\color{black}(\neg P_a,P_b,S)} \arrow[r, no head] & {\color{blue}(b,2)\color{black}(P_b,\neg S)}
	\end{tikzcd}$$\caption{The output of applying the action in Figure \ref{3Figure33} to the input in Figure \ref{3Figure32}, reprsenting a single step of the restricted case of the ``Layered Message Passing Protocol''.}\label{3Figure34}
\end{figure}

If we apply the action model to the input a single time, we get the drawing in Figure \ref{3Figure34}, which is as given on page 31 of \cite{DCTCT}. If we apply it again, we get the output as drawn in Figure \ref{3Figure35}, which is again as desired. 

\begin{figure}
	$$\adjustbox{scale=.85,center}{\begin{tikzcd}
			{\color{blue}((b,4),2)\color{black}(\neg P_a,P_b,\neg S)} \arrow[r, no head] \arrow[rd, no head] & {\color{red}((a,3),3)\color{black}(\neg P_a,P_b,S)} \arrow[r, no head] & {\color{blue}((b,4),4)\color{black}(\neg P_a,P_b,S)} \arrow[r, no head] & {\color{red}((a,3),5)\color{black}(P_b,\neg P_a,\neg S)} \arrow[ld, no head] \\
			& {\color{red}((a,5),1)\color{black}(\neg P_a,\neg S)}               & {\color{blue}((b,2),6)\color{black}(P_b,\neg S)}               &                                                                      
	\end{tikzcd}}$$\caption{The output of applying the action in Figure \ref{3Figure33} to the input in Figure \ref{3Figure34}, reprsenting the second step of the restricted case of the ``Layered Message Passing Protocol''.}\label{3Figure35}
\end{figure}

Now we want to apply this intuition to a more complicated setting, where the bit values are not predetermined in the input. This input is drawn in Figure \ref{3Figure36}.

\begin{figure}
	$$\begin{tikzcd}
		{\color{blue}b_1\color{black}(P_b,S)} \arrow[d, no head] \arrow[r, no head] & {\color{red}a_0\color{black}(\neg P_a,S)} \arrow[d, no head] \\
		{\color{red}a_1\color{black}(P_a,S)} \arrow[r, no head]                     & {\color{blue}b_0\color{black}(\neg P_b,S)}                  
	\end{tikzcd}$$\caption{The full input for the ``Layered Message Passing Protocol''.}\label{3Figure36}
\end{figure}
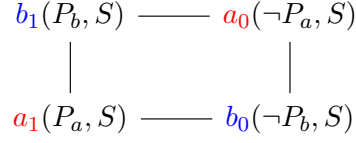

Delightfully, the action model is easy to describe given the work we've already done. We simply need four copies of the action model previously, corresponding with the four combinations of pairs of bit values. We then stitch these together and we get Figure \ref{3Figure37}.

\begin{figure}
	$$\adjustbox{scale=.85,center}{\begin{tikzcd}
			{\color{blue}(1,4)\color{black}(P_b,S)[\neg S]} \arrow[r, no head] \arrow[rd, no head] \arrow[d, no head] \arrow[rrdd, no head, bend left] & {\color{red}(2,4)\color{black}(\neg P_a,P_b,S)[P_b]} \arrow[r, no head]                & {\color{blue}(3,4)\color{black}(\neg P_a,P_b,S)[\neg P_a]} \arrow[r, no head] & {\color{red}(4,4)\color{black}(\neg P_a,S)[\neg S]} \arrow[ld, no head] \arrow[d, no head] \arrow[lldd, no head, bend left] \\
			{\color{red}(1,3)\color{black}(P_a,P_b,S)[P_b]} \arrow[d, no head]                                                                         & {\color{red}(2,3)\color{black}(\neg P_a,\neg S)} \arrow[r, no head] \arrow[d, no head] & {\color{blue}(3,3)\color{black}(P_b,\neg S)} \arrow[d, no head]     & {\color{blue}(4,3)\color{black}(\neg P_a,\neg P_b,S)[\neg P_a]} \arrow[d, no head]                                                    \\
			{\color{blue}(1,2)\color{black}(P_a,P_b,S)[P_a]} \arrow[d, no head]                                                                        & {\color{blue}(2,2)\color{black}(\neg P_b,\neg S)}                                      & {\color{red}(3,2)\color{black}(P_a,\neg S)}                         & {\color{red}(4,2)\color{black}(\neg P_a,\neg P_b,S)[\neg P_b]} \arrow[d, no head]                                                     \\
			{\color{red}(1,1)\color{black}(P_a,S)[\neg S]} \arrow[ru, no head] \arrow[r, no head] \arrow[rruu, no head, bend left]                     & {\color{blue}(2,1)\color{black}(P_a,\neg P_b,S)[P_a]} \arrow[r, no head]               & {\color{red}(3,1)\color{black}(P_a,\neg P_b,S)[\neg P_b]} \arrow[r, no head]  & {\color{blue}(4,1)\color{black}(\neg P_b,S)[\neg S]} \arrow[lu, no head] \arrow[lluu, no head, bend left]                  
	\end{tikzcd}}$$\caption{The full action for the ``Layered Message Passing Protocol''.}\label{3Figure37}
\end{figure}
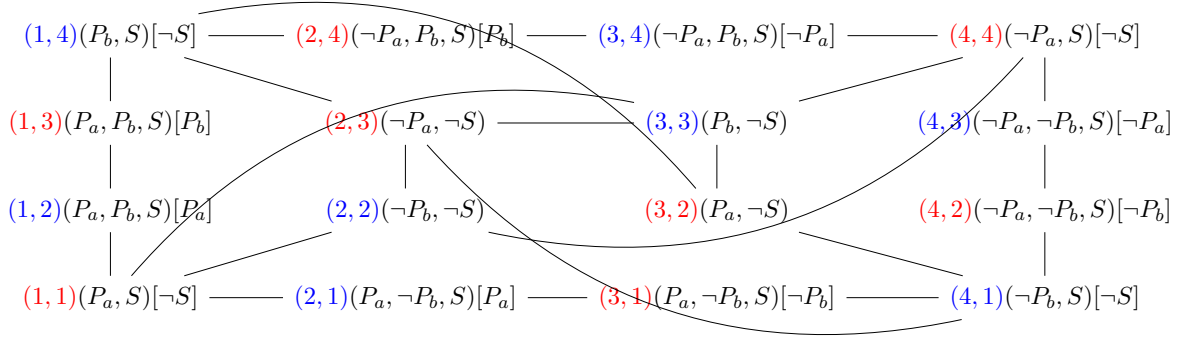

Applying the action model once gives us the model in Figure \ref{3Figure38}, exactly as we would want per page 31 of \cite{DCTCT}. Applying it a second time gives us the model in Figure \ref{3Figure39}. The values of variables are left off for ease of reading. It is here that it is important we include the non-ideal additional pre- and post- conditions. Otherwise, we could have perspectives in the above like $(a_0,(2,4),(4,2))$, which would be where $a_0$ first learns $P_b$ and then $\neg P_b$. Nevertheless, we can show that the protocol applied $n$ times is a square with sidelength $2n-1$:

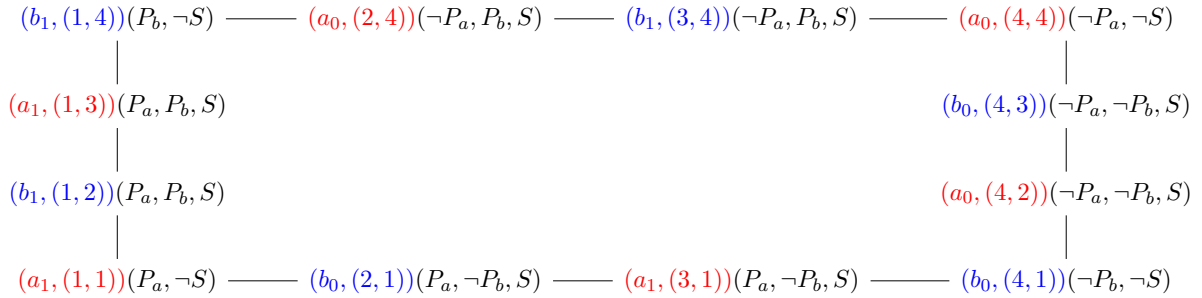
\begin{figure}
	$$\adjustbox{scale=.85,center}{\begin{tikzcd}
			{\color{blue}(b_1,(1,4))\color{black}(P_b,\neg S)} \arrow[r, no head] \arrow[d, no head] & {\color{red}(a_0,(2,4))\color{black}(\neg P_a,P_b,S)} \arrow[r, no head]  & {\color{blue}(b_1,(3,4))\color{black}(\neg P_a,P_b,S)} \arrow[r, no head] & {\color{red}(a_0,(4,4))\color{black}(\neg P_a,\neg S)} \arrow[d, no head] \\
			{\color{red}(a_1,(1,3))\color{black}(P_a,P_b,S)} \arrow[d, no head]                      &                                                                      &                                                                      & {\color{blue}(b_0,(4,3))\color{black}(\neg P_a,\neg P_b,S)} \arrow[d, no head] \\
			{\color{blue}(b_1,(1,2))\color{black}(P_a,P_b,S)} \arrow[d, no head]                     &                                                                      &                                                                      & {\color{red}(a_0,(4,2))\color{black}(\neg P_a,\neg P_b,S)} \arrow[d, no head]  \\
			{\color{red}(a_1,(1,1))\color{black}(P_a,\neg S)} \arrow[r, no head]                     & {\color{blue}(b_0,(2,1))\color{black}(P_a,\neg P_b,S)} \arrow[r, no head] & {\color{red}(a_1,(3,1))\color{black}(P_a,\neg P_b,S)} \arrow[r, no head]  & {\color{blue}(b_0,(4,1))\color{black}(\neg P_b,\neg S)}                  
	\end{tikzcd}}$$\caption{The output of applying the action in Figure \ref{3Figure37} to the input in Figure \ref{3Figure36}, representing the a single step of the full ``Layered Message Passing Protocol''. }\label{3Figure38}
\end{figure}

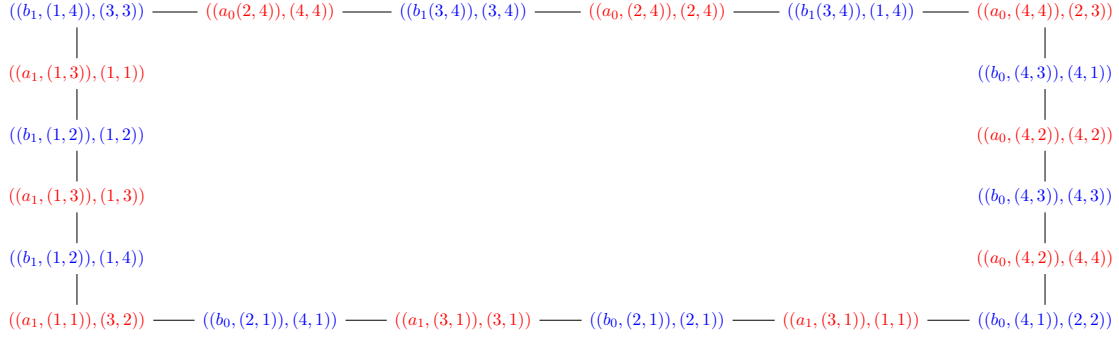
\begin{figure}
	$$\adjustbox{scale=.6,center}{\begin{tikzcd}
			{\color{blue}((b_1,(1,4)),(3,3))} \arrow[r, no head] & {\color{red}((a_0(2,4)),(4,4))} \arrow[r, no head]   & {\color{blue}((b_1(3,4)),(3,4))} \arrow[r, no head] & {\color{red}((a_0,(2,4)),(2,4))} \arrow[r, no head]  & {\color{blue}((b_1(3,4)),(1,4))} \arrow[r, no head] & {\color{red}((a_0,(4,4)),(2,3))} \arrow[d, no head]  \\
			{\color{red}((a_1,(1,3)),(1,1))} \arrow[u, no head]  &                                                      &                                                     &                                                      &                                                     & {\color{blue}((b_0,(4,3)),(4,1))} \arrow[d, no head] \\
			{\color{blue}((b_1,(1,2)),(1,2))} \arrow[u, no head] &                                                      &                                                     &                                                      &                                                     & {\color{red}((a_0,(4,2)),(4,2))} \arrow[d, no head]  \\
			{\color{red}((a_1,(1,3)),(1,3))} \arrow[u, no head]  &                                                      &                                                     &                                                      &                                                     & {\color{blue}((b_0,(4,3)),(4,3))} \arrow[d, no head] \\
			{\color{blue}((b_1,(1,2)),(1,4))} \arrow[u, no head] &                                                      &                                                     &                                                      &                                                     & {\color{red}((a_0,(4,2)),(4,4))} \arrow[d, no head]  \\
			{\color{red}((a_1,(1,1)),(3,2))} \arrow[u, no head]  & {\color{blue}((b_0,(2,1)),(4,1))} \arrow[l, no head] & {\color{red}((a_1,(3,1)),(3,1))} \arrow[l, no head] & {\color{blue}((b_0,(2,1)),(2,1))} \arrow[l, no head] & {\color{red}((a_1,(3,1)),(1,1))} \arrow[l, no head] & {\color{blue}((b_0,(4,1)),(2,2))} \arrow[l, no head]
	\end{tikzcd}}$$\caption{The output of applying the action in Figure \ref{3Figure37} to the input in Figure \ref{3Figure38}, representing the second step of the full ``Layered Message Passing Protocol''. }\label{3Figure39}
\end{figure}

The above examples are our base case. Suppose Output(n) is a square of sidelength $2(n+1)$. The center two nodes on each edge are labeled $S$, all other nodes $\neg S$, and every node but the corners contains a unique $a_i,b_j$ pair. Applying the action duplicates every node but the centers of the edges of the input. These are duplicated by the center four nodes of the action, namely $(2,3)$, $(3,3)$, $(2,2)$, and $(3,2)$. The two center nodes of each edge are first copied by their respective $a_i,b_j$ pair. These copies will be the center two nodes of the string of four nodes that will replace each center pair. In this case, the $a_0$ $b_1$ edge are copied by $(2,4)$, $(3,4)$. Then the $a$ is copied additionally by the $a_i$ corner, and the $b$ node is copied additionally by the $b_j$ corner. These are the outer two nodes of that string of four nodes. No other copies are possible, as these nodes specify both the $a_i$ and $b_j$ value, and are $S$ nodes. So, as desired, we have a single copy of every node except the center of each edge, which gets two copies each, as desired.

\subsection{Incorporating Revision} \label{3sec:Rev}

There are two definitions. There's updating by a belief action, and revising by a belief action.

\begin{definition}[Simplicial Action for Belief]
	A \defin{Simplicial Action for Belief} $A$ is a tuple $\langle N_A,V_A,L_A,\{S_{a,A}\}_{a\in Ag},Post\rangle$, where $Post:N_a\rightarrow 3^\mathfrak{P}$.\footnote{$N_A$ is a set of nodes, $V_A:N_A\rightarrow Ag$, $L_A:N_A\rightarrow 3^\mathfrak{P}$, and each $S_{a,A}$ is a UCF subcomplex of $\mathfrak{M}(N_A,V_A,L_A)$, just as in a regular simplicial model.}
	
	Given a simplicial model for belief $\mathcal{M}=\langle N,V,L,S\rangle$, And a simplicial action for belief $A=\langle N_A,V_A,L_A,\{S_{a,A}\}_{a\in Ag},Post\rangle$, we define the update model $\mathcal{M}[A]_B:=\langle N[A]_B,V[A]_B,L[A]_B,\{S_a[A]_B\}_{a\in Ag}\rangle$\footnote{$N[A]_B$ is a subset of $N\times N_A$, $V[A]_B:N[A]_B\rightarrow Ag$, $L[A]_B:N[A]_B\rightarrow 3^\mathfrak{P}$, and $S_a[A]_B$ is a UCF subcomplex of $\mathfrak{M}(N[A]_B,V[A]_B,L[A]_B)$, just as in a regular simplicial model.} as follows:
	
	\begin{align*}
		N[A]_B&:=\{(x,y)\in N\times N_A|V(x)=V_A(y)\\
		&\wedge\forall P\in\mathfrak{P}((L(x)(P)=1\rightarrow L_A(y)(P)\neq 0)\wedge(L(x)(P)=0\rightarrow L_A(y)(P)\neq 1))\}
		\\ V[A]_B((x,y))&:=V(x)
		\\ L[A]_B((x,y))(P)&:=\begin{cases}
			& Post(y)(P)\;\text{if}\;Post(y)(P)\neq 2 \\
			& L_A(y)(P)\;\text{if}\;L(x)(P)=2 \\
			&L(x)\;\text{otherwise}
		\end{cases} 
		\\ \mathcal{F}(S_a[A]_B)&:=\left \{F\in 2^{N[A]_B}~\middle |~\begin{matrix}
			& \pi_0(F)\in\mathcal{F}(S_a) \\
			& \wedge\pi_1(F)\in\mathcal{F}(S_{a,A}) \\
			& \wedge F\text{ is consistent w.r.t. }L[A]_B \\
		\end{matrix} \right \}
	\end{align*}
\end{definition}

We should check that the update model in this setting is UCF.

\begin{theorem}\label{3thm:simpbelUCF} Let $\M$ be a simplicial model for belief and $A$ a simplicial action model for belief. Then the output model $\mathcal{M}[A]_B$ also satisfies the UCF property.
	\begin{proof}
		Proof is almost identical to knowledge case. Showing each $S_a[A]_B$ is UCF is the same as showing $S[A]$ is UCF. More detail in Section \ref{3prf:simpbelUCF}.
	\end{proof}
\end{theorem}

We would like to say that this notion of an action is a protocol. However, the notion of protocol we used previously, the one given in \cite{DCTCT}, is insufficient for our purposes. Intersection preservation is meant to ensure that no agent ``knows'' whether they came from face $X$ or $Y$. Since we have shifted over to belief, that is no longer available to us. We need the corresponding version of the intersection property which uses the definition of truth for modalities in the belief setting. In particular, we will need an execution map for each agent. But, it will only be relevant that the $a$-colored execution map preserve the intersections of $a$-colored vertices.\footnote{This is especially relevant in the revision setting.} This motivates the following definition of a \textbf{Belief Protocol}:

\begin{definition}[Belief Protocol]
	A \defin{Belief Protocol} is a triple $(\mathcal{I},\mathcal{P},\{\Xi_a\}_{a\in Ag})$ where $\mathcal{I}$ is a simplicial model for belief, called the \defin{``Input Complex''}, and $\mathcal{P}$ is a simplicial model for belief called the \defin{``Protocol Complex''} (better thought of in our context as the output complex). Let $\mathcal{I}_a$ denote the $a$-complex for $\mathcal{I}$ and similarly for $\mathcal{P}_a$. For each $a\in Ag$ a function $\Xi_a:\mathcal{I}\rightarrow2^\mathcal{P}$ is called the \defin{``Execution Map''}, which takes faces in $\mathcal{I}_a$ to sets of faces in $\mathcal{P}_a$ such that if $\sigma\subseteq\tau$, $\Xi_a(\sigma)\subseteq\Xi_a(\tau)$, and $\Xi_a(\pi_a(\sigma\cap\sigma'))=\pi_a(\Xi_a(\sigma))\cap\pi_a(\Xi_a(\sigma'))$. Moreover, for any face $\sigma$ in $\mathcal{I}_a$, $\{a\in Ag~|~\exists x\in\sigma(V_\mathcal{I}(x)=a)\}=\{a\in Ag~|~\exists X\in\Xi_a(\sigma),x\in X(V_\mathcal{P}(x)=a)\}$. When this is true we say that $\Xi_a$ is \defin{chromatic}. And, additionally, for each $a\in Ag$, $\mathcal{P}_a=\bigcup_{\sigma\in\mathcal{I}}\Xi_a(\sigma)$.
\end{definition}

We now show the sense in which action models for belief are protocols. Let $I$ be a simplicial model for belief whose $a$ complex is denoted $I_a$, and $A$ a simplicial action for belief. As before, intuitively, our protocol complex should be given by $I[A]_B$. What we need to define is for each $a\in Ag$ an execution map $\Xi_a:I_a\rightarrow 2^{I_a[A]_B}$. As it turns out, the correct map is given by the analogous map to the knowledge case:

$$\Xi_a(F):=\{X\in I_a[A]_B~|~\pi_0(X)\subseteq F\}$$

Where $X\in I_a[A]_B$ is shorthand for ``$X$ is a face in the simplicial complex of $I_a[A]_B$.'' This gives us the following:

\begin{theorem}\label{3thm:simpbelprot} Let $I$ be a simplicial model for belief and $A$ an action model. The tuple $(I,I[A]_B,\{\Xi_a\}_{a\in Ag})$ is a Belief Protocol.
	\begin{proof}
		Section \ref{3prf:simpbelprot}
	\end{proof}
\end{theorem}

Morally, this is true because of the following. Chromaticity follows roughly because $V[A]_B((x,y))=V(x)$. This then entails the intersection property roughly as follows. If $F$ and $G$ share an $a$-node, call it $x$, $\Xi_a$ will send that $a$-node to a collection of $a$-nodes in $I[A]_B$. In particular, $\Xi_a(\{x\})$ is the set of all singletons $\{(x,y)\}$ in $I_a[A]_B$. Now consider $\Xi_a(F)$. By definition, this is the set of all faces $X$ in $I_a[A]_B$ such that $\pi_0(X)\subseteq F$. Since $\pi_0(X)=\{z~|~(z,y)\in X\}$, if $X\in\Xi_a(F)$, then $\pi_a(\pi_0(X))=\{x\}$. Hence, $\Xi_a(F)$ contains the set of all faces $X$ in $I_a[A]_b$ whose $a$-node is a tuple $(x,y)$. Hence, $\pi_a(\Xi_a(F))$ is the set of all singletons $\{(x,y)\}$ in $I_a[A]_B$. The same reasoning applies to $G$ as to $F$. So, we can conclude that 

$$\Xi_a(\pi_a(F\cap G))=\pi_a(\Xi_a(F))=\pi_a(\Xi_a(G))$$

as desired.

Now we need to modify this definition of an action to use the notion of revision from \cite{SimpBelRev}. The key idea is to modify the facets from $\mathcal{F}(S_a[A]_B)$. 

\begin{definition}[Simplicial Action for Belief with Revision]
	Given a simplicial model for belief $\mathcal{M}=\langle N,V,L,S\rangle$, and a simplicial action for belief $A=\langle N_A,V_A,L_A,\{S_{a,A}\}_{a\in Ag},Post\rangle$, we define the update model \defin{with revision} $\mathcal{M}[A]_{BR}:=\langle N[A]_{BR},V[A]_{BR},L[A]_{BR},\{S_a[A]_{BR}\}_{a\in Ag}\rangle$\footnote{$N[A]_B$ is a subset of $N\times N_A$, $V[A]_{BR}:N[A]_{BR}\rightarrow Ag$, $L[A]_{BR}:N[A]_{BR}\rightarrow 3^\mathfrak{P}$, and $S_a[A]_{BR}$ is a UCF subcomplex of $\mathfrak{M}(N[A]_{BR},V[A]_{BR},L[A]_{BR})$, just as in a regular simplicial model.} as follows:
	
	\begin{align*}
		N[A]_{BR}&:=\{(x,y)\in N\times N_A|V(x)=V_A(y)\\
		&\wedge\forall P\in\mathfrak{P}((L(x)(P)=1\rightarrow L_A(y)(P)\neq 0)\wedge(L(x)(P)=0\rightarrow L_A(y)(P)\neq 1))\}
		\\ V[A]_{BR}((x,y))&:=V(x)
		\\ L[A]_{BR}((x,y))(P)&:=\begin{cases}
			& Post(y)(P)\;\text{if}\;Post(y)(P)\neq 2 \\
			& L_A(y)(P)\;\text{if}\;L(x)(P)=2 \\
			&L(x)\;\text{otherwise}
		\end{cases} 
	\end{align*}
	
	Defining $S_a[A]_{BR}$ will take some more work. Fix $a\in Ag$ and let $X$ be a UCF facet in $2^{N\times N_A}$ such that $\pi_0(X)\in\mathcal{F}(S_a)$ and $\pi_1(X)\in\mathcal{F}(S_{a,A})$. Call the set of such $X$ $\mathcal{F}^a(N\times N_A)$. Let $X$ be a UCF facet in $2^{N[A]_{BR}}$ such that $X$ is consistent, and $\pi_1(X)\in\mathcal{F}(S_{a,A})$. Call the set of such $X$ $\mathcal{F}_a(2^{N[A]_{BR}})$. 
	
	Then we can define a replacement function as follows for $X\in\mathcal{F}^a(N\times N_A)$:
	
	$$R_a(X):=\{Y\in\mathcal{F}_a(2^{N[A]_{BR}})|\pi_a(Y)=\pi_a(X)$$
	$$\wedge(\forall Z\in\mathcal{F}_a(2^{N[A]_{BR}})(\pi_a(Z)=\pi_a(X)\rightarrow|Y\cap X|\geq|Z\cap X|))\}$$
	
	We say that $S_a[A]_{BR}$ is defined as the simplicial complex whose facets come from $\bigcup_{X\in\mathcal{F}^a(N\times N_A)}R_a(X)$.
\end{definition}

\begin{theorem}\label{3thm:simpbelrevUCF} Let $\M$ be a simplicial model for belief and $A$ a simplicial action model for belief. Then the output model $\mathcal{M}[A]_{BR}$ also satisfies the UCF property.
	\begin{proof}
		Section \ref{3prf:simpbelrevUCF}
	\end{proof}
\end{theorem}

Now we need to show that simplicial actions for belief with revision are belief protocols. To do this, we need the following notion. Note that if $Y\in\F(S_a[A]_{BR})$, then $R_a^{-1}(Y)$ is the preimage. We need to define this preimage for all faces in $S_a[A]_BR$. Given a face $X\in S_a[A]_{BR}$, say that $R_a^{-1}(X)$ is the set of faces which subset of all facets

\begin{align*}
	R_a^{-1}(X):=\{X'~|~&\exists Y\in\F(S_a[A]_{BR})(X\subseteq Y\wedge\exists Y'\in R_a^{-1}(Y)(X'\subseteq Y'))\\&\wedge\forall a\in Ag(\exists x\in X(V[A]_{BR}(x)=a)\leftrightarrow \exists x'\in X'(V[A]_{BR}(x')=a))\}
\end{align*}

\begin{theorem}\label{3thm:simpbelrevprot} Let $I$ be a simplicial model for belief and $A$ an action model. Define $\Xi_{a,BR}$ similarly to $\Xi_a$:
	
	$$\Xi_{a,BR}(F):=\{X\in I_a[A]_{BR}~|~\forall X'\in R_a^{-1}(X)(\pi_0(X')\subseteq F)\}$$
	
	Then, the tuple $(I,I[A]_{BR},\{\Xi_{a,BR}\}_{a\in Ag})$ is a Belief Protocol.
	\begin{proof}
		Section \ref{3prf:simpbelrevprot}
	\end{proof}
\end{theorem}

Morally, this is true because the replacement function takes facets to sets of facets that, among other things, share $a$-colored nodes. Hence, the same moral argument we gave above still applies to show the intersection property.

\section{Advantages of Action Models}

In the introduction for this paper, we suggested that it would be interesting to explore connecting action models and simplicial semantics since both have tight connections to distributed computing. Now that we have done so, we can talk more concretely about specific advantages afforded by modeling protocols using simplicial actions. The first concerns the epistemic intuitiveness of actions. In all of the examples above, we were able to create the formal object of an action for a specific protocol simply by carefully considering an agent's sender and receiver behavior. As a result, modifying actions becomes quite simple. For example, consider the action model drawn in Figure \ref{3FigureClassicAction}.\footnote{Note that this is the simplicial translation of the action given on page 22 of \cite{sus}.} This action describes ``$a$ sends $P_a$ to $b$ with a chance of failure, and $c$ is unaware that any message is sent.'' This is easy to read off of the action. The perspective $a_S$ is $a$'s perspective where they send the message. They do not know if $b$ will receive it, hence they consider a facet possible where $b$ is in perspective $b_R$, where the message is received, and another $b_{\sim R}$, where it is not. When the message is received, $b$ knows this, and so there is only one facet in $S_b$ containing $b_R$. Agent $c$, by contrast, only considers a single world possible, namely one where no message was sent. They also believe every other agent only considers this world possible. This is the facet $\{a_{\sim S}, b_{\sim R},c\}$, which is in all three of $S_a$, $S_b$, and $S_c$.

\begin{figure}
	$$\begin{tikzcd}
		{\color{red} a_{\color{black}S}}\color{black}(P_a) \arrow[blue, rd, no head, bend right] \arrow[blue, rr, no head, bend left] \arrow[red, dd, no head] \arrow[red, rd, no head] \arrow[red, rr, no head] &                                                                                                                & \color{blue} b_{\color{black}R}\color{black}[P_a] \\
		& \color{green} c\color{black}(\top) \arrow[rd, no head] \arrow[red, ru, no head] \arrow[blue, ru, no head, bend right] &                                 \\
		{\color{blue} b_{\color{black}\sim R}}\color{black}(\top) \arrow[ru, no head] \arrow[rr, no head]                                                                                   &                                                                                                                & \color{red} a_{\color{black}\sim S}\color{black}(\top) 
	\end{tikzcd}$$\caption{The simplicial action model describing ``$a$ sends $P_a$ to $b$ with a chance of failure, and $c$ is unaware that any message is sent.''}\label{3FigureClassicAction}
\end{figure}
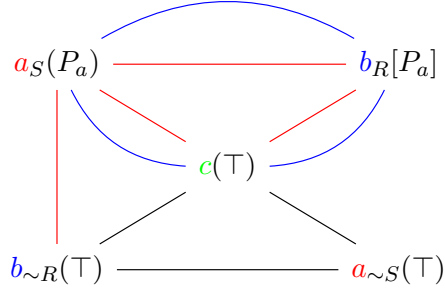

Suppose we wanted to modify this protocol so that $c$ was aware that a message containing $P_a$ may have been sent, but they are unsure, and if it was sent, they are unsure that it was received. Modifying Figure \ref{3FigureClassicAction} to account for this is straightforward. Agent $c$ simply considers more facets possible. Namely, the facets where $P_a$ was sent and received, and also where it was sent but not received. This is done by adding $\{a_S,b_R,c\}$ and $\{a_S,b_{\sim R}, c\}$ to $\mathcal{F}(S_c)$. The result is drawn in Figure \ref{3FigureClassModAct}.

\begin{figure}
	$$\begin{tikzcd}
		{\color{red} a_{\color{black}S}}\color{black}(P_a) \arrow[red, dd, no head] \arrow[green, dd, no head, bend right] \arrow[black, rd, no head] \arrow[green, rd, no head, bend left] \arrow[red, rd, no head, bend right] \arrow[black, rr, no head] &                                                                                                                & \color{blue} b_{\color{black}R}\color{black}[P_a] \\
		& \color{green} c\color{black}(\top) \arrow[rd, no head] \arrow[black, ru, no head] &                                 \\
		{\color{blue} b_{\color{black}\sim R}}\color{black}(\top) \arrow[ru, no head] \arrow[green, ru, no head, bend right] \arrow[red, ru, no head, bend left] \arrow[rr, no head]                                                                                   &                                                                                                                & \color{red} a_{\color{black}\sim S}\color{black}(\top) 
	\end{tikzcd}$$\caption{The simplicial action model describing ``$a$ sends $P_a$ to $b$ with a chance of failure, and $c$ thinks it's possible that the message was successfully sent, unsuccessfully sent, or no message was sent at all.''}\label{3FigureClassModAct}
\end{figure}
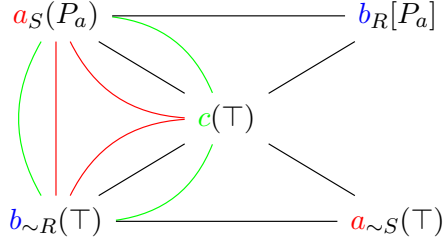

Another reason describing protocols as actions is advantageous is that protocols disentangle the ``sender/receiver'' behavior from the input model. In the usual definition of protocol, repeated above, the input/output complexes are specified, and the protocol is given as a morphism specifying which facets map where.\cite{DCTCT} Hence, the protocol can't be distentangled from a particular input.\footnote{This is also true of the definition of protocol given in \cite{TopPersp}, though we will say more about this presentation below.} By contrast, actions can be applied to any input model with a sensible result as the output. For example, the actions drawn in Figure \ref{3FigureClassicAction} and Figure \ref{3FigureClassModAct} could be applied to any input whose vertices correspond to six perspectives where each of three agents has an assigned bit value. For example, these actions could be applied to the model drawn in Figure \ref{3fgr:SimpBelEx1}, or any UCF complex drawn on those same vertices.

Another advantage of action models is that, in addition to separating out ``signals'' from ``actions'', they also compartmentalize information in a convenient way. Consider again the ``Alternating Message Passing Protocol''. In the language of \cite{DCTCT}, this is a ``layered'' protocol, as it transforms over many steps. So, to represent it in the traditional way, we require an infinite sequence of complexes. We would start with the drawing in Figure \ref{3Figure26}, then the drawing in Figure \ref{3Figure31}, followed by another drawing with the branches extended a step further, and so on. This sequence is also drawn on page 30 of \cite{DCTCT}. However, we are able to represent this protocol \textit{entirely} with two simplicial models, rather than an infinite sequence. The input, given in Figure \ref{3Figure26}, and the action, given in \ref{3Figure30}, suffice. As we proved, any action specifies a protocol. So, we can describe at least some of these more complicated layered protocols with simply an input and an action, which we understand to be applied at each stage, rather than with an infinite sequence of models.

This disentangling of input and protocol affords more opportunities for testing protocols as well. In the near future, we are hoping to be able to implement the application of simplicial actions to inputs in a language like Python, in order to be able to consider more and larger examples more easily. Since the application of actions to inputs generates very large models very quickly, this would seem to be essential for further research. However, in such a setting, the separation of inputs and protocols would appear to be an advantage. Given an appropriate specification of vertices $N$, $V$, and $L$, perhaps a protocol is ``successful'' on (given a concrete definition of successful) \textit{most}, but not all, UCF subcomplexes of $\mathfrak{M}(N,V,L)$. We could then evaluate if cases where the protocol fails are sufficiently pathological as to be avoided or considered irrelevant. This testing environment would also allow for the ability to test belief revision itself in a philosophically exciting way. Given that applying a simplicial belief action with or without revision generates a belief protocol, we should be able to directly compare the performance of various protocols with and without revision. This would let us evaluate in an empirical sense when and how revision gives learners an advantage, which would be an exciting philosophical argument.

The last advantage of the action model presentation of protocols is that we might be able to leverage the topological properties of simplicial action models in the description of protocols. Leveraging such topological properties is a core part of the application of combinatorial topology to distributed computing. \cite{DCTCT} However, ours is not the first example of a topological characterization of protocols themselves (not just the inputs and outputs). The paper ``A topological perspective on distributed network algorithms'', for example, has already done this. \cite{TopPersp} We will briefly recount a core example from that paper. Consider the knowledge complex drawn in Figure \ref{3FigureTopPersp1}. Here, agents have a privately held bit value, and there is total uncertainty. The goal will be for agents to send messages in such a sequence that they could plausibly come to consensus. There is no need for the majority to win, either. If two agents initially have $0$, but ultimately all agents agree to flip to 1, this is considered a success.\footnote{However, if all agents have $0$, they must agree on $0$ in the end.} Hence, the desired output complex consists of two disconnected facets, drawn in Figure \ref{3FigureTopPersp2}.

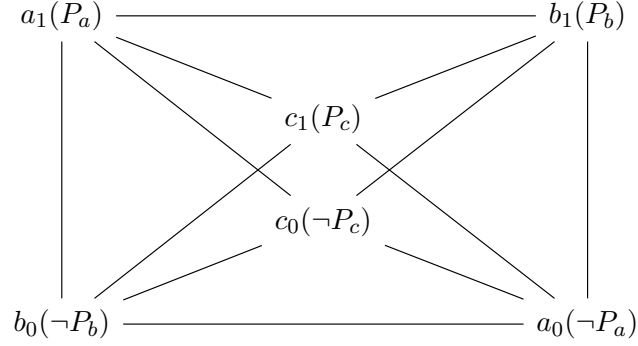
\begin{figure}
	$$\begin{tikzcd}
		a_1(P_a) \arrow[rrd, no head] \arrow[rrdd, no head] \arrow[ddd, no head] &  &                          &  & b_1(P_b) \arrow[llll, no head] \arrow[lldd, no head] \arrow[ddd, no head] \\
		&  & c_1(P_c) \arrow[rru, no head] &  &                                                                      \\
		&  & c_0(\neg P_c) \arrow[lld, no head] &  &                                                                      \\
		b_0(\neg P_b) \arrow[rrrr, no head] \arrow[rruu, no head]                     &  &                          &  & a_0(\neg P_a) \arrow[llu, no head] \arrow[lluu, no head]                      
	\end{tikzcd}$$\caption{Knowledge complex for the initial configuration of three agents, each with a private bit value, with total uncertainty.}\label{3FigureTopPersp1}
\end{figure}

\begin{figure}
	$$\begin{tikzcd}
		a_1(P_a) \arrow[rrd, no head]  &  &                          &  & b_1(P_b) \arrow[llll, no head] \\
		&  & c_1(P_c) \arrow[rru, no head] &  &                           \\
		&  & c_0(\neg P_c) \arrow[lld, no head] &  &                           \\
		b_0(\neg P_b) \arrow[rrrr, no head] &  &                          &  & a_0(\neg P_a) \arrow[llu, no head] 
	\end{tikzcd}$$\caption{Action describing ``$a$ sends their bit value to both $b$ and $c$.''}\label{3FigureTopPersp2}
\end{figure}
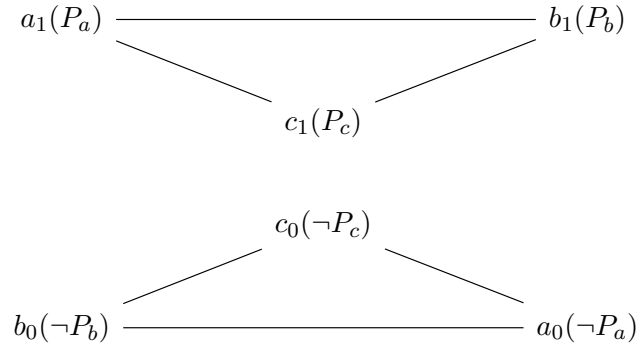

In \cite{TopPersp}, they consider two different protocols as potential solutions to this problem. The first involves a cycle of messages: $a$ signals $b$ her value, $b$ signals $c$ her value, and $c$ signals $a$ her value. It is very important to note that they consider protocols in the ``know-all'' setting. Hence, all agents are aware that this is the protocol, and there is no uncertainty. This is called the ``structure awareness'' assumption. Moreover, they effectively describe protocols geometrically as what we would call the application of a protocol to an action. That is, for an input $I$ and an action $A$, their definition of a ``protocol complex'' is effectively what we would call $I[A]$, though they do not separately specify the actions. Hence, they define a ``protocol complex'' as the result of applying a sequence of signals to an input. Moreover, because of the ``structural awareness'' assumption, there must be a bijection between the facets of the input and the facets of the protocol complex. This is because they assume that there is common knowledge of what the input worlds are, what the messages are, who sends them, and who receives them. So every agent knows exactly what messages will be sent and received at any given world. Obviously, action model descriptions of protocols, as given above, have no such restrictions, even in merely the knowledge setting.

The result of applying the cyclical sequence of messages to Figure \ref{3FigureTopPersp1} is drawn in Figure \ref{3FigureTopPersp3}. Each perspective is split, depending on whether they heard a $1$ or a $0$. However, crucially, the complex remains connected.

\begin{figure}
	$$\adjustbox{scale=.9,center}{\begin{tikzcd}
			& {a_1(P_a,P_c)} \arrow[rrd, no head] \arrow[rrdd, no head]             &  &                                                                                            &  & {b_1(P_b,P_a)} \arrow[llll, no head] \arrow[llllld, no head]         &                                                                  \\
			{a_1(P_a,\neg P_c)} \arrow[rrrdd, no head] \arrow[ddd, no head] &                                                                       &  & {c_1(P_c,P_b)} \arrow[rru, no head] \arrow[rrrddd, no head]                                &  &                                                                      & {b_1(P_b,\neg P_a)} \arrow[lll, no head] \arrow[llldd, no head]  \\
			&                                                                       &  & {c_1(P_c,\neg P_b)} \arrow[rrrdd, no head] \arrow[llldd, no head]                          &  &                                                                      &                                                                  \\
			&                                                                       &  & {c_0(\neg P_c,P_b)} \arrow[rrdd, no head] \arrow[rruuu, no head]                           &  &                                                                      &                                                                  \\
			{b_0(\neg P_b,P_a)} \arrow[ruuuu, no head]                      &                                                                       &  & {c_0(\neg P_c,\neg P_b)} \arrow[lld, no head] \arrow[llluuu, no head] \arrow[lll, no head] &  &                                                                      & {a_0(\neg P_a,P_c)} \arrow[uuu, no head] \arrow[llllld, no head] \\
			& {b_0(\neg P_b,\neg P_a)} \arrow[rrrr, no head] \arrow[rruuu, no head] &  &                                                                                            &  & {a_0(\neg P_a,\neg P_c)} \arrow[llu, no head] \arrow[ruuuu, no head] &                                                                 
	\end{tikzcd}}$$\caption{Knowledge complex for the protocol where $a$ signals $b$, $b$ signals $c$, and $c$ signals $a$.}\label{3FigureTopPersp3}
\end{figure}
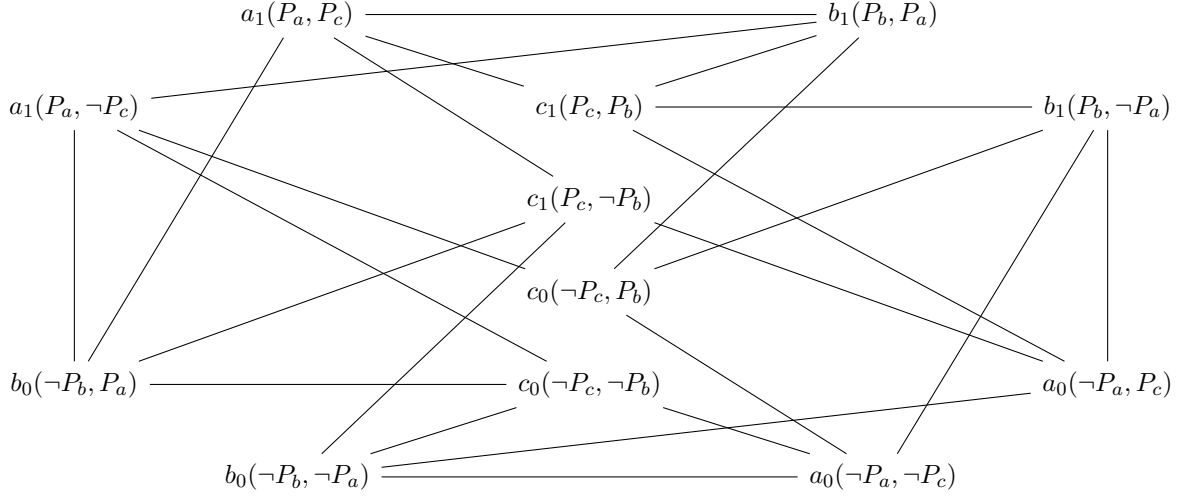

By contrast, consider instead the sequence of signals ``$a$ signals both $b$ and $c$''. Agent $a$'s perspective would not duplicate, as they learn nothing. But agent $b$'s and agent $c$'s both do, as they learn $a$'s bit value. This result of this is drawn in Figure \ref{3FigureTopPersp4}. Crucially, this complex is disconnected. Because the target, too, is disconnected, it is necessary for the analysis given in \cite{TopPersp} that the protocol complex be disconnected as well. Hence, the first cyclical sequence of signals does not solve the task of consensus, but $a$ broadcasting does. This application of connectivity leverages this topological depiction of protocols.

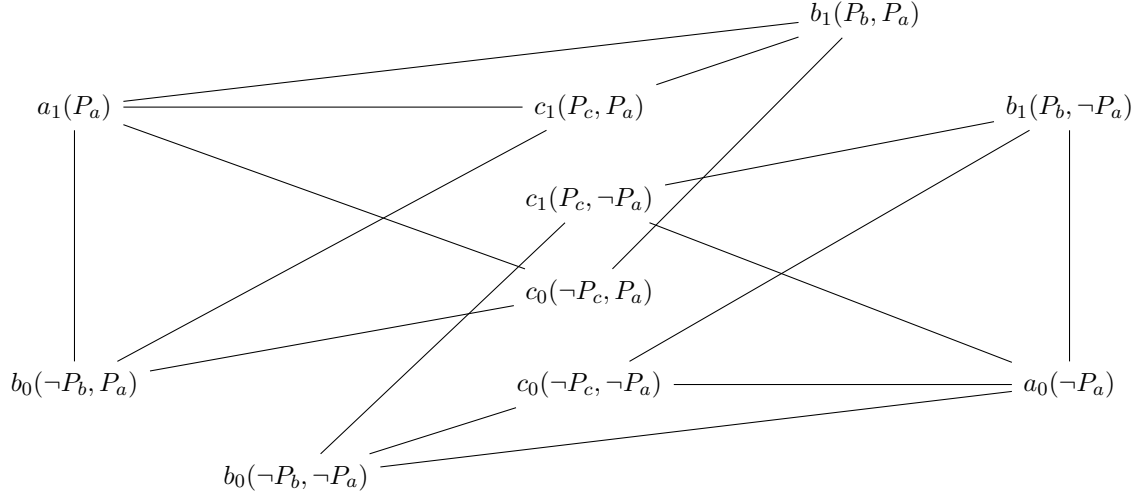
\begin{figure}
	$$\adjustbox{scale=.9,center}{\begin{tikzcd}
			&                                                 &  &                                                                                            &  & {b_1(P_b,P_a)} \arrow[llllld, no head] \arrow[llddd, no head] &                                                            \\
			a_1(P_a) \arrow[ddd, no head] \arrow[rrr, no head] &                                                 &  & {c_1(P_c,P_a)} \arrow[rru, no head] \arrow[lllddd, no head]                                &  &                                                               & {b_1(P_b,\neg P_a)}                                        \\
			&                                                 &  & {c_1(P_c,\neg P_a)} \arrow[rrru, no head] \arrow[rrrdd, no head]                           &  &                                                               &                                                            \\
			&                                                 &  & {c_0(\neg P_c,P_a)} \arrow[llld, no head] \arrow[llluu, no head]                           &  &                                                               &                                                            \\
			{b_0(\neg P_b,P_a)}                                &                                                 &  & {c_0(\neg P_c,\neg P_a)} \arrow[lld, no head] \arrow[rrruuu, no head] \arrow[rrr, no head] &  &                                                               & a_0(\neg P_a) \arrow[uuu, no head] \arrow[llllld, no head] \\
			& {b_0(\neg P_b,\neg P_a)} \arrow[rruuu, no head] &  &                                                                                            &  &                                                               &                                                           
	\end{tikzcd}}$$\caption{Knowledge complex for the protocol where $a$ signals $b$ and $c$.}\label{3FigureTopPersp4}
\end{figure}

As mentioned before, \cite{TopPersp} is viewing ``protocol complexes'' as what we would call the application of an action to an input. So, let's draw these two protocol complexes instead as actions. Consider first the cycle, where $a$ signals $b$, $b$, signals $c$, and $c$ signals $a$. Considering both sender and receiver behavior, each agent has $4$ perspectives: they either send a 1 or 0, and they either receive a 1 or 0. In fact, there is no difference between our action and the protocol complex drawn in Figure \ref{3FigureTopPersp3}, save for how we label postconditions. This is drawn in Figure \ref{3Figure40}. One can check that the application of this action to the input is trivial, as it just redraws the action. Moreover, the application of this action to the input in Figure \ref{3FigureTopPersp1} does yield the model given in Figure \ref{3FigureTopPersp4}, which would be as desired.

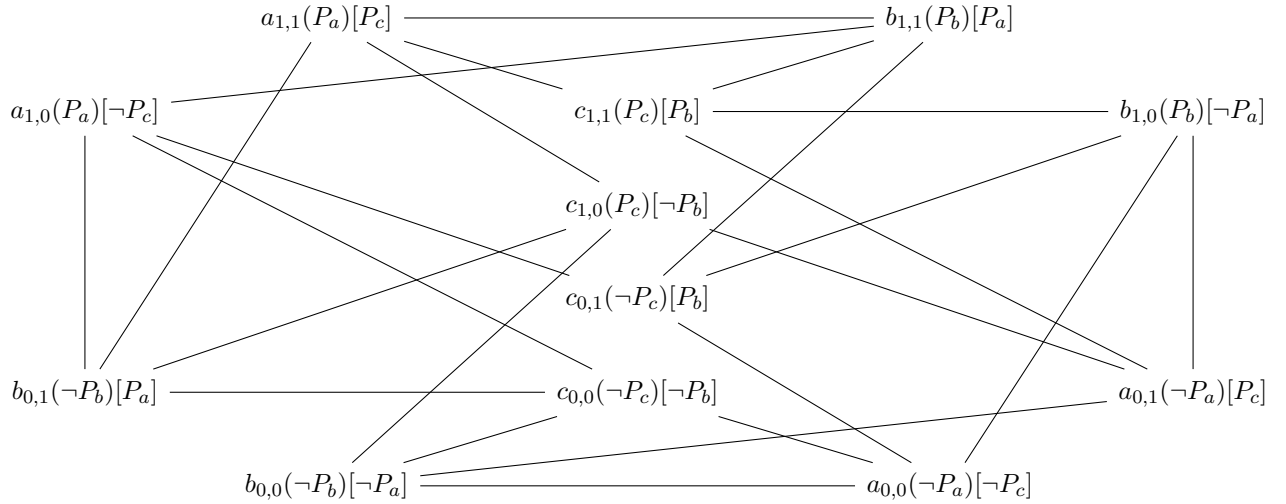
\begin{figure}
	$$\adjustbox{scale=.9,center}{\begin{tikzcd}
			& {a_{1,1}(P_a)[P_c]} \arrow[rrd, no head] \arrow[rrdd, no head]             &  &                                                                                            &  & {b_{1,1}(P_b)[P_a]} \arrow[llll, no head] \arrow[llllld, no head]         &                                                                  \\
			{a_{1,0}(P_a)[\neg P_c]} \arrow[rrrdd, no head] \arrow[ddd, no head] &                                                                       &  & {c_{1,1}(P_c)[P_b]} \arrow[rru, no head] \arrow[rrrddd, no head]                                &  &                                                                      & {b_{1,0}(P_b)[\neg P_a]} \arrow[lll, no head] \arrow[llldd, no head]  \\
			&                                                                       &  & {c_{1,0}(P_c)[\neg P_b]} \arrow[rrrdd, no head] \arrow[llldd, no head]                          &  &                                                                      &                                                                  \\
			&                                                                       &  & {c_{0,1}(\neg P_c)[P_b]} \arrow[rrdd, no head] \arrow[rruuu, no head]                           &  &                                                                      &                                                                  \\
			{b_{0,1}(\neg P_b)[P_a]} \arrow[ruuuu, no head]                      &                                                                       &  & {c_{0,0}(\neg P_c)[\neg P_b]} \arrow[lld, no head] \arrow[llluuu, no head] \arrow[lll, no head] &  &                                                                      & {a_{0,1}(\neg P_a)[P_c]} \arrow[uuu, no head] \arrow[llllld, no head] \\
			& {b_{0,0}(\neg P_b)[\neg P_a]} \arrow[rrrr, no head] \arrow[rruuu, no head] &  &                                                                                            &  & {a_{0,0}(\neg P_a)[\neg P_c]} \arrow[llu, no head] \arrow[ruuuu, no head] &                                                                 
	\end{tikzcd}}$$\caption{Action model for the protocol where $a$ signals $b$, $b$ signals $c$, and $c$ signals $a$.}\label{3Figure40}
\end{figure}

The action for the other protocol is much more interesting by contrast. There are two $a$ perspectives, depending on what $a$ sends. For $b$ and $c$ each, there are also two perspectives, depending on what they receive. Since $a$ does not receive anything, and $b$ and $c$ do not send anything, this fully characterizes the set of perspectives. This is drawn in Figure \ref{3Figure41}. Note that this complex is disconnected, just as the target is. So, possibly, we can leverage topology similar to how it is done in \cite{TopPersp}. However, much more work needs to be done to generalize this.

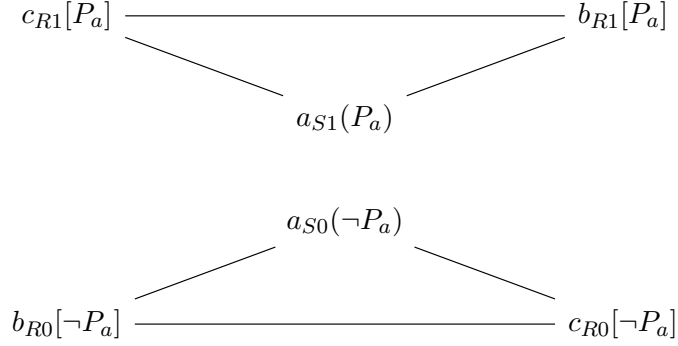
\begin{figure}
	$$\begin{tikzcd}
		{c_{R1}[P_a]} \arrow[rrd, no head]       &  &                                       &  & {b_{R1}[P_a]} \arrow[llll, no head]     \\
		&  & a_{S1}(P_a) \arrow[rru, no head]      &  &                                         \\
		&  & a_{S0}(\neg P_a) \arrow[lld, no head] &  &                                         \\
		{b_{R0}[\neg P_a]} \arrow[rrrr, no head] &  &                                       &  & {c_{R0}[\neg P_a]} \arrow[llu, no head]
	\end{tikzcd}$$\caption{Action model for the protocol where $a$ signals $b$ and $c$.}\label{3Figure41}
\end{figure}

However, it should be noted that action models do seem to be far more general than the definition of protocol considered in \cite{TopPersp}. As mentioned above, the facets of a protocol complex must be in bijection with the input.\cite{TopPersp} As we have seen, our action model definition of protocol has no such restriction. This is in part because we are not beholden to the same ``awareness assumption'' as they are in that paper. Hence, agents can be uncertain of whether or not a particular message was sent, or of its contents, or of who the sender or receiver is, and so no. In belief protocols, agents can even have a false conception of what messages are being sent. These generalizations directly expand on the work in \cite{TopPersp}. Consider the following passage from the conclusion of \cite{TopPersp}, where they consider future work:

\begin{quote}``For instance, fully capturing the popular local model requires removing the structure awareness assumption, and studying the details of how the protocol complex evolves round after round.''\cite{TopPersp}\end{quote}

It is possible that, by analyzing protocols using actions, and in particular considering belief protocols, we have presented a model with sufficient generality to capture the ``local model''. However, more work would need to be done to check this. However, it is clear that we have generalized the conception of protocol given in \cite{TopPersp}. Moreover, the action presentation of protocols still equips them with a topology, which might be able to be leveraged in familiar ways, such as what is briefly explored above, and possibly novel ways we are not yet aware of. Following up on this is an exciting dimension for future research.

\section{Conclusion} \label{3sec:conc}

In this paper, we set out to give a new definition for action models in the setting of simplicial semantics. In particular, we wanted a notion of action which could apply to models for belief as in \cite{BelSimp}, so that we could incorporate the concept of revision from \cite{SimpBelRev}. As a result, we are able to, on the one hand, describe many existing, albeit simple, simplicial protocols as action models. Forthcoming work with the NASA Langley Research Center also uses these definitions to define a novel protocol. This protocol is a gossip protocol in the sense of \cite{Birman07}, and using this we are able to explore the possibilities of using revision, as defined here, as a further means of fault-tolerance in distributed protocols. However, exploring the use of revision as a means of fault-tolerance is still ongoing work.

This leads us to directions of future work. The first major direction for future work would be to implement action models in simplicial complexes in a simulation. In particular, a program which could compute action updates would be very helpful. Even small examples get large very quickly, and being able to work with more examples would help us to answer many questions, some of which we have already posed.

Firstly, does belief revision actually help with fault tolerance? If so, how? In what contexts? Answering this would require comparing and contrasting many protocols, ideally being able to compare protocols who only differ in whether or not they use the revision updating procedure. Getting clear on exactly how revision affects a protocol, especially after many stages, would also be philosophically exciting and would be a means of empirically testing the revision in a way that could speak to its value as a model of learning.

Secondly, what are the limitations of representing protocols using actions? A question this paper has not addressed is the actual, provable, scope of using action models to represent protocols. Are there protocols which cannot be represented using action models? If so, which ones? Why? A place to start would seem to be to take the definition of protocol from \cite{DCTCT} and try to see if one can prove that all such protocols could be defined as actions.

\bibliographystyle{plain}	
\bibliography{Epistemology} 

@book{DCTCT,
	author = {Herlihy, Maurice and Kozlov, Dmitry and Rajsbaum, Sergio},
	title = {Distributed Computing Through Combinatorial Topology},
	year = {2013},
	isbn = {0124045782},
	publisher = {Morgan Kaufmann Publishers Inc.},
	address = {San Francisco, CA, USA},
	edition = {1st}
}

@article{SimpDEL,
	title = "A simplicial complex model for dynamic epistemic logic to study distributed task computability",
	author = "{\'E}ric Goubault and J{\'e}r{\'e}my Ledent and Sergio Rajsbaum",
	year = "2021",
	month = jun,
	doi = "10.1016/j.ic.2020.104597",
	language = "English",
	volume = "278",
	journal = "Information and Computation",
	issn = "0890-5401",
	publisher = "Elsevier Inc.",
	
}

@InProceedings{KaSC,
	author="Ditmarsch, Hans van
	and Goubault, {\'E}ric
	and Ledent, J{\'e}r{\'e}my
	and Rajsbaum, Sergio",
	editor="Lundgren, Bj{\"o}rn
	and Nu{\~{n}}ez Hern{\'a}ndez, Nancy Abigail",
	title="Knowledge and Simplicial Complexes",
	booktitle="Philosophy of Computing",
	year="2022",
	publisher="Springer International Publishing",
	address="Cham",
	pages="1--50",
	isbn="978-3-030-75267-5"
}

@incollection{Death,
	volume = {219},
	month = {March},
	series = {Leibniz International Proceedings in Informatics, LIPIcs},
	booktitle = {39th International Symposium on Theoretical Aspects of Computer Science (STACS 2022)},
	editor = {Petra Berenbrink and Benjamin Monmege},
	title = {A simplicial model for KB4n : epistemic logic with agents that may die},
	address = {FRA},
	publisher = {Schloss Dagstuhl - Leibniz-Zentrum f{\"u}r Informatik},
	year = {2022},
	doi = {10.4230/LIPIcs.STACS.2022.33},
	pages = {33:1---33:20},
	url = {https://doi.org/10.4230/LIPIcs.STACS.2022.33},
	isbn = {9783959772228},
	author = {Goubault, {\'E}ric and Ledent, J{\'e}r{\'e}my and Rajsbaum, Sergio}
}

@InProceedings{DoA,
	author="van Ditmarsch, Hans",
	editor="Silva, Alexandra
	and Wassermann, Renata
	and de Queiroz, Ruy",
	title="Wanted Dead or Alive: Epistemic Logic for Impure Simplicial Complexes",
	booktitle="Logic, Language, Information, and Computation",
	year="2021",
	publisher="Springer International Publishing",
	address="Cham",
	pages="31--46",
	isbn="978-3-030-88853-4"
}

@book{DC5,
	author = {Fagin, Ronald and Halpern, Joseph Y. and Moses, Yoram and Vardi, Moshe Y.},
	title = {Reasoning About Knowledge},
	year = {2003},
	isbn = {0262562006},
	publisher = {MIT Press},
	address = {Cambridge, MA, USA}
}

@InProceedings{HG,
	author =	{Goubault, \'{E}ric and Kniazev, Roman and Ledent, J\'{e}r\'{e}my},
	title =	{{A Many-Sorted Epistemic Logic for Chromatic Hypergraphs}},
	booktitle =	{32nd EACSL Annual Conference on Computer Science Logic (CSL 2024)},
	pages =	{30:1--30:18},
	series =	{Leibniz International Proceedings in Informatics (LIPIcs)},
	ISBN =	{978-3-95977-310-2},
	ISSN =	{1868-8969},
	year =	{2024},
	volume =	{288},
	editor =	{Murano, Aniello and Silva, Alexandra},
	publisher =	{Schloss Dagstuhl -- Leibniz-Zentrum f{\"u}r Informatik},
	address =	{Dagstuhl, Germany},
	URL =		{https://drops.dagstuhl.de/entities/document/10.4230/LIPIcs.CSL.2024.30},
	URN =		{urn:nbn:de:0030-drops-196730},
	doi =		{10.4230/LIPIcs.CSL.2024.30}
}

@article{FA,
	title={Simplicial Models for the Epistemic Logic of Faulty Agents}, 
	author={{\'E}ric Goubault and Roman Kniazev and J{\'e}r{\'e}my Ledent and Sergio Rajsbaum},
	year={2024},
	journal = {Bol. Soc. Mat. Mex},
	volume    = {30},
	note = {Article number: 90}
}

@INPROCEEDINGS{SimpSet,
	
	author={Goubault, Éric and Kniazev, Roman and Ledent, Jérémy and Rajsbaum, Sergio},
	
	booktitle={2023 38th Annual ACM/IEEE Symposium on Logic in Computer Science (LICS)}, 
	
	title={Semi-Simplicial Set Models for Distributed Knowledge}, 
	
	year={2023},
	
	volume={},
	
	number={},
	
	pages={1-13},
	
	doi={10.1109/LICS56636.2023.10175737}}

@Inbook{action,
	author="Baltag, Alexandru
	and Moss, Lawrence S.
	and Solecki, S{\l}awomir",
	editor="Arl{\'o}-Costa, Horacio
	and Hendricks, Vincent F.
	and van Benthem, Johan",
	title="The Logic of Public Announcements, Common Knowledge, and Private Suspicions",
	bookTitle="Readings in Formal Epistemology: Sourcebook",
	year="2016",
	publisher="Springer International Publishing",
	address="Cham",
	pages="773--812",
	isbn="978-3-319-20451-2",
	doi="10.1007/978-3-319-20451-2_38",
	url="https://doi.org/10.1007/978-3-319-20451-2_38"
}

@article{sus,
	author = {Baltag, Alexandru},
	title = {A Logic for Suspicious Players: Epistemic Actions and Belief–Updates in Games},
	journal = {Bulletin of Economic Research},
	volume = {54},
	number = {1},
	pages = {1-45},
	doi = {https://doi.org/10.1111/1467-8586.00138},
	url = {https://onlinelibrary.wiley.com/doi/abs/10.1111/1467-8586.00138},
	eprint = {https://onlinelibrary.wiley.com/doi/pdf/10.1111/1467-8586.00138},
	year = {2002}
}

@InProceedings{ActionDC,
	author="Pfleger, Daniel
	and Schmid, Ulrich",
	editor="Lotker, Zvi
	and Patt-Shamir, Boaz",
	title="On Knowledge and Communication Complexity in Distributed Systems",
	booktitle="Structural Information and Communication Complexity",
	year="2018",
	publisher="Springer International Publishing",
	address="Cham",
	pages="312--330",
	isbn="978-3-030-01325-7"
}

@article{KaG,
	author = {Attamah, M. and Ditmarsch, Hans and Grossi, D. and Hoek, Wiebe},
	year = {2014},
	month = {01},
	pages = {21-26},
	title = {Knowledge and Gossip},
	volume = {263},
	journal = {Frontiers in Artificial Intelligence and Applications},
	doi = {10.3233/978-1-61499-419-0-21}
}

@article{CPM,
	title={Communication Pattern Models: An Extension of Action Models for Dynamic-Network Distributed Systems},
	volume={335},
	ISSN={2075-2180},
	url={http://dx.doi.org/10.4204/EPTCS.335.29},
	DOI={10.4204/eptcs.335.29},
	journal={Electronic Proceedings in Theoretical Computer Science},
	publisher={Open Publishing Association},
	author={Velázquez, Diego A. and Castañeda, Armando and Rosenblueth, David A.},
	year={2021},
	month=jun, pages={307–321} }

@article{CPaAM,
	title={Comparing the Update Expressivity of Communication Patterns and Action Models},
	volume={379},
	ISSN={2075-2180},
	url={http://dx.doi.org/10.4204/EPTCS.379.14},
	DOI={10.4204/eptcs.379.14},
	journal={Electronic Proceedings in Theoretical Computer Science},
	publisher={Open Publishing Association},
	author={Castañeda, Armando and van Ditmarsch, Hans and Rosenblueth, David A. and Velázquez, Diego A.},
	year={2023},
	month=jul, pages={157–172} }

@article{BRDC1,
	title = {Conditional Doxastic Models: A Qualitative Approach to Dynamic Belief Revision},
	journal = {Electronic Notes in Theoretical Computer Science},
	volume = {165},
	pages = {5-21},
	year = {2006},
	note = {Proceedings of the 13th Workshop on Logic, Language, Information and Computation (WoLLIC 2006)},
	issn = {1571-0661},
	doi = {https://doi.org/10.1016/j.entcs.2006.05.034},
	url = {https://www.sciencedirect.com/science/article/pii/S1571066106005111},
	author = {Alexandru Baltag and Sonja Smets}
}

@article{BRDC2,
	author = {Baltag, Alexandru and Smets, Sonja},
	year = {2008},
	month = {01},
	pages = {},
	title = {The logic of conditional doxastic actions: A theory of dynamic multi-agent belief revision}
}

@article{Birman07,
	author = {Birman, Ken},
	title = {The promise, and limitations, of gossip protocols},
	year = {2007},
	issue_date = {October 2007},
	publisher = {Association for Computing Machinery},
	address = {New York, NY, USA},
	volume = {41},
	number = {5},
	issn = {0163-5980},
	url = {https://doi.org/10.1145/1317379.1317382},
	doi = {10.1145/1317379.1317382},
	journal = {SIGOPS Oper. Syst. Rev.},
	month = oct,
	pages = {8–13},
	numpages = {6}
}

@article{BelSimp,
	title={A Semantics for Belief in Simplicial Complexes},
	volume={447},
	ISSN={2075-2180},
	url={http://dx.doi.org/10.4204/EPTCS.447.10},
	DOI={10.4204/eptcs.447.10},
	journal={Electronic Proceedings in Theoretical Computer Science},
	publisher={Open Publishing Association},
	author={Bjorndahl, Adam and Sink, Philip},
	year={2026},
	month=jun, pages={173–188} }

@InProceedings{SimpBel,
	author="Cachin, Christian
	and Lehnherr, David
	and Studer, Thomas",
	editor="Schmid, Ulrich
	and Kuznets, Roman",
	title="Simplicial Belief",
	booktitle="Structural Information and Communication Complexity",
	year="2025",
	publisher="Springer Nature Switzerland",
	address="Cham",
	pages="176--193",
	isbn="978-3-031-91736-3"
}

@book{Bjorndahl2024, place={Cambridge}, series={Cambridge Introductions to Philosophy}, title={An Introduction to Classical and Modal Logics: The Outlines of Knowledge}, publisher={Cambridge University Press}, author={Bjorndahl, Adam}, year={2024}, collection={Cambridge Introductions to Philosophy}}

@article{TopPersp,
	title = {A topological perspective on distributed network algorithms},
	journal = {Theoretical Computer Science},
	volume = {849},
	pages = {121-137},
	year = {2021},
	issn = {0304-3975},
	doi = {https://doi.org/10.1016/j.tcs.2020.10.012},
	url = {https://www.sciencedirect.com/science/article/pii/S0304397520305831},
	author = {Armando Castañeda and Pierre Fraigniaud and Ami Paz and Sergio Rajsbaum and Matthieu Roy and Corentin Travers}
}

@techreport{Memo,
	author = {Sink, Philip and Goodloe, Alwyn},
	title = {A New Semantics for Belief Revision in Simplicial Complexes},
	institution =  "NASA",
	month = "June",
	year  = "2026",
	type  = "Technical Memorandum",
	number = "NASA/TM-2026-0004926",
	address = "Langley Research Center, Hampton VA 23681-2199, USA"
}

@misc{SimpBelRev,
	title={Simplicial Semantics for Belief Revision}, 
	author={Philip Sink},
	year={2026},
	eprint={2608.13763},
	archivePrefix={arXiv},
	primaryClass={cs.LO},
	url={https://arxiv.org/abs/2608.13763}, 
}

\section{Proofs} \label{3app:prf}

\begin{theorem} Let $\M$ be a simplicial model for knowledge and $A$ a simplicial action model. Then the output model $\mathcal{M}[A]$ also satisfies the UCF property.
\end{theorem}
\begin{proof}\label{3prf:simpUCF}
	Fix $b\in Ag$, and $F\in \mathcal{F}(S[A])$. We first want to show that there is $(x,y)\in F$ such that $V[A](x,y)=b$. Note that $\pi_0(F)\in\F(S)$. Hence, since $S$ is UCF, there is $x\in\pi_0(F)$ such that $V(x)=b$. Fix $y$ such that $(x,y)\in F$. Then, by the definition of $V[A]$, $V[A](x,y)=b$, as desired.
	
	Suppose $(x,y),(x',y')\in F$ and $V[A](x,y)=V[A](x',y')=b$. We need to show that $(x,y)=(x',y')$. Note that by the definition of $V[A]$, $V(x)=V(x')=b$. Similarly, by the definition of $N[A]$, $V_A(y)=V_A(y')=b$. Because $x,x'\in\pi_0(F)\in\F(S)$, and $S$ is UCF, we have that $x=x'$. Because $y,y'\in\pi_1(F)\in\F(S_A)$, and $S_A$ is UCF, we have that $y=y'$. The result follows.
\end{proof}

\begin{theorem} Let $I$ be a simplicial model for knowledge and $A$ an action model. The tuple $(I,I[A],\Xi)$ is a Simplicial Protocol.
\end{theorem}
\begin{proof}\label{3prf:simpprot}
	Suppose $F\subseteq G$. Then suppose $X\in\Xi(F)$. Then $\pi_0(X)\subseteq F$. Then $\pi_0(X)\subseteq G$, so $X\in\Xi(G)$, as desired.
	
	Suppose $X\in\Xi(F\cap G)$. Then $\pi_0(X)\subseteq F\cap G$. So, $\pi_0(X)\subseteq F$ and $\pi_0(X)\subseteq G$, and we get that $X\in\Xi(F)\cap\Xi(G)$, as desired. The reverse is identical.
	
	In order to show that $\Xi$ is chromatic we will first establish the result for singletons. Let $n\in I$ be a node and fix $a\in Ag$ such that $V(n)=a$. Then in particular, if $X\in\Xi(\{n\})$, then $\pi_0(X)=n$. It follows from the UCF property that $|X|=1$, and fixing $(n,y)\in X$, $V[A]((n,y))=V(n)=a$. 
	
	Now we generalize this to larger faces. For each $n\in F$ where $F$ is a face in $I$, we know that $V[A](\Xi(\{n\}))=V(n)$. Suppose $X\in\Xi(F)$. Then, for each $a\in Ag$ such that there's an $(x,y)\in X$ such that $V[A]((x,y))=V(x)=a$, because $\pi_0(X)\subseteq F$, $x\in F$ and thus there is an $a$-perspective in $F$. This suffices to show chromaticity.
	
	By construction, $\bigcup_{\sigma\in\mathcal{I}}\Xi(\sigma)\subseteq I[A]$. Suppose that $X$ is a face in $I[A]$. Take $\sigma=\pi_0(X)$. By the definition of $I[A]$, $\sigma$ is a face in $I$. Fix $Y$ any facet containing $X$. This suffices to show $\bigcup_{\sigma\in\mathcal{I}}\Xi(\sigma)=I[A]$.
\end{proof}

\begin{theorem} Let $\M$ be a simplicial model for belief and $A$ a simplicial action model for belief. Then the output model $\mathcal{M}[A]_B$ also satisfies the UCF property.
\end{theorem}
\begin{proof}\label{3prf:simpbelUCF}
	Fix $b\in Ag$, and $F\in \mathcal{F}(S_a[A]_B)$. We first want to show that there is $(x,y)\in F$ such that $V[A]_B(x,y)=b$. Note that $\pi_0(F)\in\F(S_a)$. Hence, since $S_a$ is UCF, there is $x\in\pi_0(F)$ such that $V(x)=b$. Fix $y$ such that $(x,y)\in F$. Then, by the definition of $V[A]_B$, $V[A]_B(x,y)=b$, as desired.
	
	Suppose $(x,y),(x',y')\in F$ and $V[A]_B(x,y)=V[A]_B(x',y')=b$. We need to show that $(x,y)=(x',y')$. Note that by the definition of $V[A]_B$, $V(x)=V(x')=b$. Similarly, by the definition of $N[A]_B$, $V_A(y)=V_A(y')=b$. Because $x,x'\in\pi_0(F)\in\F(S_a)$, and $S_a$ is UCF, we have that $x=x'$. Because $y,y'\in\pi_1(F)\in\F(S_{a,A})$, and $S_{a,A}$ is UCF, we have that $y=y'$. The result follows.
\end{proof}

\begin{theorem} Let $I$ be a simplicial model for belief and $A$ an action model. The tuple $(I,I[A]_B,\{\Xi_a\}_{a\in Ag})$ is a Belief Protocol.
\end{theorem}
\begin{proof}\label{3prf:simpbelprot}
	Suppose $F$ and $G$ are faces of $S_a$ and $F\subseteq G$. Then suppose $X\in\Xi_a(F)$. Then $\pi_0(X)\subseteq F$. Then $\pi_0(X)\subseteq G$, so $X\in\Xi_a(G)$, as desired.
	
	In order to show that each $\Xi_{a}$ is chromatic we will first establish the result for singletons. Let $n\in I$ be a node and fix $a\in Ag$ such that $V(n)=a$. Then in particular, if $X\in\Xi_{a}(\{n\})$, then $\pi_0(X)=n$. It follows from the UCF property that $|X|=1$, and fixing $(n,y)\in X$, $V[A]((n,y))=V(n)=a$. 
	
	Now we generalize this to larger faces. For each $n\in F$ where $F$ is a face in $I$, we know that $V[A](\Xi_{a}(\{n\}))=V(n)$. Suppose $X\in\Xi_{a}(F)$. Then, for each $a\in Ag$ such that there's an $(x,y)\in X$ such that $V[A]((x,y))=V(x)=a$, because $\pi_0(X)\subseteq F$, $x\in F$ and thus there is an $a$-perspective in $F$. This suffices to show chromaticity.
	
	Let $F$ and $G$ be two faces of $S_a$. To verify the intersection property, we can see that the only interesting case is when $\pi_a(F\cap G)\neq\emptyset$. If $\pi_a(F\cap G)=\emptyset$, then, $\Xi_{a}(\pi_a(F\cap G))=\pi_a(\Xi_{a}(F))\cap\pi_a(\Xi_{a}(G))=\emptyset$. 
	Suppose $X\in\Xi_{a}(\pi_a(F\cap G))$. Then $\pi_0(X)\subseteq \pi_a(F\cap G)$. Note that this implies that $\pi_0(X)$ is a singleton containing the unique $a$-colored node in both $F$ and $G$. So, $\pi_0(X)\subseteq F$ and $\pi_0(X)\subseteq G$, and we get that $X\in\pi_a(\Xi_{a}(F))\cap\pi_a(\Xi_{a}(G))$, as desired. Suppose now that $X\in\pi_a(\Xi_{a}(F))\cap\pi_a(\Xi_{a}(G))$. Then $X\in\pi_a(\Xi_{a}(F))$. So, by the UCF property, $X$ is an $a$-colored singleton in $I_a[A]$. By the definition of $\Xi_a$, there is $X'$ such that $X\subseteq X'$ and $\pi_0(X')\subseteq F$. Since $X$ contains the unique $a$-colored node of $X'$, we can conclude that $\pi_0(X)\subseteq F$ is the singleton set containing the unique $a$-colored node of $F$. The same argument will show that $\pi_0(X)$ is the singleton set containing the unique $a$-colored node of $G$. So, $\pi_0(X)$ is the singleton set containing the unique $a$-colored node of $F\cap G$. Hence, $\pi_0(X)=\pi_a(F\cap G)$. By definition, then, $X\in\Xi_{a}(\pi_a(F\cap G))$. This suffices to show the intersection property.
	
	By construction, $\bigcup_{\sigma\in\mathcal{I}}\Xi_a(\sigma)\subseteq I[A]$. Suppose that $X$ is a face in $I[A]$. Take $\sigma=\pi_0(X)$. By the definition of $I[A]$, $\sigma$ is a face in $I$. Fix $Y$ any facet containing $X$. This suffices to show $\bigcup_{\sigma\in\mathcal{I}}\Xi_a(\sigma)=I[A]$.
\end{proof}

\begin{theorem} Let $\M$ be a simplicial model for belief and $A$ a simplicial action model for belief. Then the output model $\mathcal{M}[A]_{BR}$ also satisfies the UCF property.
\end{theorem}
\begin{proof}\label{3prf:simpbelrevUCF}
	This follows immediately from the fact that, by assumption, each facet in $S_a[A]_{BR}$ is given by $\mathcal{R}_a(X)$ for some $X\in\mathcal{F}^a(\mathcal{M}[A]_B)$. Because $\mathcal{R}_a(X)\subseteq\mathcal{F}_a(2^{N[A]_{BR}})$ for all $X\in\mathcal{F}^a(\mathcal{M}[A]_B)$, and the facets in $\mathcal{F}_a(2^{N[A]_{BR}})$ are assumed to be UCF, the result follows.
\end{proof}

\begin{theorem} Let $I$ be a simplicial model for belief and $A$ an action model. Define $\Xi_{a,BR}$ similarly to $\Xi_a$:
	
	$$\Xi_{a,BR}(F):=\{X\in I_a[A]_{BR}~|~\forall X'\in R_a^{-1}(X)(\pi_0(X')\subseteq F)\}$$
	
	Then, the tuple $(I,I[A]_{BR},\{\Xi_{a,BR}\}_{a\in Ag})$ is a Belief Protocol.
\end{theorem}
\begin{proof}\label{3prf:simpbelrevprot}
	Suppose $F$ and $G$ are faces of $S_a$ and $F\subseteq G$. Then suppose $X\in\Xi_{a,BR}(F)$. Then for all $X'\in R_a^{-1}(X)$ $\pi_0(R_a^{-1}(X'))\subseteq F$. Then $\pi_0(R_a^{-1}(X'))\subseteq G$, so $X\in\Xi_{a,BR}(G)$, as desired.
	
	In order to show that each $\Xi_{a,BR}$ is chromatic we will first establish the result for singletons. Let $n\in I$ be a node and fix $a\in Ag$ such that $V(n)=a$. Then in particular, if $X\in\Xi_{a,BR}(\{n\})$, then for all $X'\in R_a^{-1}(X)$, $\pi_0(X')=n$. It follows from the UCF property and the fact that $R_a^{-1}$ preserves colors that $|X|=|X'|=1$. Fix $y$ such that $(n,y)\in X'$ and $x$ and $z$ such that $(x,z)\in X$. By the definition of $R_a^{-1}$, $V[A]_{BR}((x,z))=V[A]_{BR}((n,y))=V(n)=a$. 
	
	Now we generalize this to larger faces. For each $n\in F$ where $F$ is a face in $I$, we know that $V[A]_{BR}(\Xi_{a,BR}(\{n\}))=V(n)$. Suppose $X\in\Xi_{a,BR}(F)$. Consider each $a\in Ag$ such that there's an $(x,y)\in X$ such that $V[A]_{BR}((x,y))=V(x)=a$. For all $X'\in R_a^{-1}(X)$, we have that $\pi_0(X')\subseteq F$. Fix $(z,w)\in X'$. Then $z\in F$ and because $a=V(x)=V[A]_{BR}(x,y)=V[A]_{BR}(z,w)=V(z)$, there is an $a$-perspective in $F$. This suffices to show chromaticity.
	
	Let $F$ and $G$ be two faces of $S_a$. To verify the intersection property, we can see that the only interesting case is when $\pi_a(F\cap G)\neq\emptyset$. If $\pi_a(F\cap G)=\emptyset$, then, by the singleton case of chromaticity, it follows that $\Xi_{a,BR}(\pi_a(F\cap G))=\pi_a(\Xi_{a,BR}(F))\cap\pi_a(\Xi_{a,BR}(G))=\emptyset$. 
	
	It will be useful for the reasoning below to realize the following. If $X$ is a singleton, $a$-colored node, then $R_a^{-1}(X)=X$. Suppose $X'\in R_a^{-1}(X)$. Then by definition there is $Y\in\F(S_a[A]_{BR})$ such that $X\subseteq Y$ and there exists $Y'\in R_a^{-1}(Y)$ where $X'\subseteq Y'$. By the definition of $R_a$, $\pi_a(Y)=\pi_a(Y')$. So, by UCF, and the fact that $R_a^{-1}$ preserves colors, $X'=\pi_a(Y')$ and $X=\pi_a(X')$.
	=		
	Suppose $X\in\Xi_{a,BR}(\pi_a(F\cap G))$. Then for all $X'\in R_a^{-1}(X)$, $\pi_0(X')\subseteq \pi_a(F\cap G)$. Note that this implies that $\pi_0(X')$ is a singleton containing the unique $a$-colored node in both $F$ and $G$. So, $\pi_0(X')\subseteq F$ and $\pi_0(X')\subseteq G$, and we get that $X\in\pi_a(\Xi_{a,BR}(F))\cap\pi_a(\Xi_{a,BR}(G))$, as desired. Suppose now that $X\in\pi_a(\Xi_{a,BR}(F))\cap\pi_a(\Xi_{a,BR}(G))$. Then $X\in\pi_a(\Xi_{a,BR}(F))$. So, by the UCF property, $X$ is an $a$-colored singleton. By the fact that $R_a^{-1}$ preserves color, for all $X'\in R_a^{-1}(X)$, $X'$ is an $a$-colored singleton. By the definition of $\Xi_{a,BR}$, $\pi_0(X')\subseteq F$. Since $X'$ contains the unique $a$-colored node of $X$, we can conclude that $\pi_0(X')\subseteq F$ is the singleton set containing the unique $a$-colored node of $F$. The same argument will show that $\pi_0(X')$ is the singleton set containing the unique $a$-colored node of $G$. So, $\pi_0(X')$ is the singleton set containing the unique $a$-colored node of $F\cap G$. Hence, $\pi_0(X')=\pi_a(F\cap G)$. By definition, then, $X\in\Xi_{a,BR}(\pi_a(F\cap G))$. This suffices to show the intersection property.
	
	By construction, $\bigcup_{\sigma\in\mathcal{I_a}}\Xi_{a,BR}(\sigma)\subseteq I_a[A]_{BR}$. Suppose that $X$ is a face in $I[A]_{BR}$. We first need to show that $R_a^{-1}(X)$ is nonempty. Fix $Y\in\mathcal{F}(S_a[A]_{BR})$ such that $X\subseteq Y$. Then by definition, there is $Y'\in R_a^{-1}(Y)$. Fix $X'$ to be the subset of $Y'$ which shares the same colors as $X$. Then $X'\in R_a^{-1}(X)$, as desired. Because $Y'\in\F^a(N\times N_A)$, $\pi_0(Y')\in\F(I_a)$. So, because $X'\subseteq Y'$, $\pi_0(X')\in I_a$. It follows that $X\in\Xi_{a,BR}(\pi_0(X'))$. This suffices to show $\bigcup_{\sigma\in\mathcal{I}}\Xi_{a,BR}(\sigma)=I[A]_{BR}$.
\end{proof}
	
\end{document}